\documentclass[journal,comsoc]{IEEEtran}

\usepackage{graphicx}
\usepackage{amsmath,amsfonts,amsthm,amssymb, bm}
\usepackage{enumitem,kantlipsum,cite}
\usepackage{tikz}
\usetikzlibrary{shapes,arrows,positioning,calc}
\newtheorem{proposition}{Proposition}
\theoremstyle{remark}
\newtheorem{remark}{Remark}
\newtheorem{theorem}{Theorem}

\usepackage[linesnumbered,ruled,vlined]{algorithm2e}
\newcommand{\T}{\mathrm{Toep}}

\usepackage[T1]{fontenc}

\ifCLASSINFOpdf
\else
\fi
\usepackage{amsmath}
\usepackage[cmintegrals]{newtxmath}
\title{Maximum Likelihood-based Gridless DoA Estimation Using Structured Covariance Matrix Recovery and SBL with Grid Refinement}
\begin{document}
%
\author{Rohan~R.~Pote,~\IEEEmembership{Student~Member,~IEEE,}
        Bhaskar~D.~Rao,~\IEEEmembership{Fellow,~IEEE}
\thanks{This research was supported by ONR Grant No. N00014-18-1-2038, NSF grant CCF-2225617, and the UCSD Center for Wireless Communications.}
\thanks{The authors are with the Department of Electrical and Computer Engineering, University of California, San Diego, La Jolla, CA 92092, USA (rpote@eng.ucsd.edu, brao@eng.ucsd.edu).}
}
\maketitle

\begin{abstract}
    In this work, we consider the parametric data model employed in applications such as line spectral estimation and direction-of-arrival estimation. We focus on the stochastic maximum likelihood estimation (MLE) framework and offer approaches to estimate the parameter of interest in a \emph{gridless} manner, overcoming the model complexities of the past. This progress is enabled by the modern trend of reparameterization of the objective, and exploiting the sparse Bayesian learning (SBL) approach. The latter is shown to be a correlation-aware method, and for the underlying problem it is identified as a grid-based technique for recovering a structured covariance matrix of the measurements. For the case when the structured matrix is expressible as a sampled Toeplitz matrix, such as when measurements are sampled in time or space at regular intervals, additional constraints and reparameterization of the SBL objective leads to the proposed structured matrix recovery technique based on MLE. The proposed optimization problem is non-convex, and we propose a majorization-minimization based iterative procedure to estimate the structured matrix; each iteration solves a semidefinite program. We recover the parameter of interest in a gridless manner by appealing to the Carath\'eodory-F\'ejer result on decomposition of positive semidefinite (PSD) Toeplitz matrices. For the general case of irregularly spaced time or spatial samples, we propose an iterative SBL procedure that refines grid points to increase resolution near potential source locations, while maintaining a low per iteration complexity. We provide numerical results to evaluate and compare the performance of the proposed techniques with other gridless techniques, and the Cram\'er-Rao bound. The proposed correlation-aware approach is more robust to environmental/system effects such as low number of snapshots, correlated sources, small separation between source locations and improves sources identifiability.   
\end{abstract}
\begin{IEEEkeywords}
Sparse signal recovery, maximum likelihood, sparse Bayesian learning, gridless estimation, correlation-aware, structured matrix recovery, correlated sources, grid refinement
\end{IEEEkeywords}
\section{Introduction}
Consider the following parametric data model\begin{equation}
    \mathbf{y}_l=\mathbf{\Phi}_{\bm{\theta}}\mathbf{x}_l+\mathbf{n}_l,\qquad 0\leq l<L,\label{eq:parmodl}
\end{equation}where $\mathbf{y}_l\in\mathbb{C}^M$ denotes the measurements, and $L$ denotes the total number of snapshots available. The $k$th column of $\mathbf{\Phi}_{\bm{\theta}}\in\mathbb{C}^{M\times K}$ is a vector function of the parameter $\theta_k$ i.e., $[\bm{\Phi}_{\bm{\theta}}]_k=\bm{\phi}(\theta_k)$ for some known $\bm{\phi}(.),k\in\{1,\ldots,K\}$. $\bm{\theta}=[\theta_1,\ldots,\theta_K]^T$ and $\theta_k$'s lie in some known continuous domain. $K$ denotes the number of sources. The
sources' signal $\mathbf{x}_l\in\mathbb{C}^K$ and noise $\mathbf{n}_l\in\mathbb{C}^M$ are independent of each other, and i.i.d. over time. The noise, $\mathbf{n}_l$, is 
distributed as
$\mathcal{CN}(\mathbf{0},\sigma^2_n\mathbf{I})$. In (\ref{eq:parmodl}), the parameters $(\bm{\theta},\mathbf{x}_l,\sigma_n^2)$ are the unknowns. The model parameters affect the measurements in a non-linear manner, which makes the inverse problem extremely difficult to solve, even in the absence of noise. The above problem is ubiquitous, with applications including biomagnetic imaging \cite{gorodnitsky95}, functional approximations \cite{natarajan95}, and echo cancellation \cite{duttweiler00}. In this work we are concerned with problems such as in line spectral estimation\cite{stoicabook05} and direction-of-arrival (DoA) estimation\cite{krim96} for narrowband signals; we emphasize the latter as means for exposition.\setlength{\textfloatsep}{1pt}\begin{table}
    \centering
    \begin{tabular}{c}
    \hspace{-0.2cm}\begin{tabular}{c@{\hskip -0.35cm}|c}
        \hline
        \bf{Methods} & \bf{Primary Bottleneck}\cite{krim96}\\\hline
        \begin{tabular}{ll}
        \hspace{-0.3cm}(a) & \hspace{-0.55cm}\begin{tabular}{l}
            i. Spatial filtering (beamforming) \\
            ii. Subspace based methods
        \end{tabular}\end{tabular} & \hspace{-0.6cm}\begin{tabular}{l}
            Aperture/ degrees of freedom\\
            \emph{Number of snapshots}
        \end{tabular}\\\hline
        \begin{tabular}{ll}
        \hspace{-0.5cm}(b) & \hspace{-0.5cm}\begin{tabular}{l}
             Deterministic/ Stochastic MLE
        \end{tabular}\end{tabular} & \hspace{-0.05cm}Model \& computational complexity\\\hline
    \end{tabular}\medskip\\
    (a) Spectral based methods (b) Parametric methods\end{tabular}
    \caption{}\vspace{-0.3cm}\label{tab:traditional}
\end{table} Approaches to solve (\ref{eq:parmodl}) have a rich history and can broadly be classified as traditional vs. modern, both significant in insights and contributions.\\
\noindent\emph{On traditional approaches}: They can be further classified into spectral based\cite{capon69,schmidt86,roy89} and parametric methods\cite{krim96}. The typical ingredients to solve (\ref{eq:parmodl}) include geometrical properties (e.g., subspace orthogonality in MUltiple SIgnal Classification (MUSIC)\cite{schmidt86} or Estimation of Signal Parameters via Rotational Invariance Techniques (ESPRIT)\cite{roy89}) and statistical properties of the model in (\ref{eq:parmodl}). A common thread that unites these methods is the usage of the second order statistics of the data. A second order statistic offers benefits such as a) compact representation of the data when\footnote{This condition was rightfully pointed out by a reviewer.} $L\geq M$ (also, sample covariance matrix serves as a sufficient statistic when data is Gaussian distributed) b) model based interpretation of data with much fewer parameters. Parametric methods are particularly attractive as they do not suffer from the bottlenecks faced by beamforming and subspace based methods (summary in Table~\ref{tab:traditional})
. Parametric methods like maximum likelihood estimation (MLE) allow one to introduce meaningful parameters as a means to incorporate information about geometry and prior, which may be inferred even with a \emph{single snapshot}. The main issues with MLE methods are the model complexity, as the resulting cost function may be highly non-linear in the parameters to solve, and often the model order is unknown.
\\\noindent\emph{On modern approaches}: These techniques, under the rubric of sparse signal recovery (SSR), involve a) reparameterization of the original problem in (\ref{eq:parmodl}) b) explicit or implicit sparsity regularization and corresponding optimization problem. They recover the parameter of interest in either grid-based or grid-\emph{less} manner, and most often explicitly impose sparsity. Under the grid-based reparameterization\cite{mallat93,chen98,rao99,tropp06,malioutov05,wipf04}, the methods first discretize the possible values of $\theta$ and introduce the measurement matrix $\mathbf{\Phi}\in\mathbb{C}^{M\times G}$, $G$ denotes the grid size. The $i$-th column $[\mathbf{\Phi}]_i=\bm{\phi}(\theta_i),i=1,\ldots,G$, and $M\ll G$. The original problem in (\ref{eq:parmodl}) can be re-written as\begin{equation}
    \mathbf{y}_l=\mathbf{\Phi}\bar{\mathbf{x}}_l+\bar{\mathbf{n}}_l,\qquad 0\leq l<L,\label{eq:sblparmodl}
\end{equation}where it is known that $\bar{\mathbf{X}}=[\bar{\mathbf{x}}_0,\ldots,\bar{\mathbf{x}}_{L-1}]$ is row-sparse i.e., most of the rows are zero. The problem in (\ref{eq:sblparmodl}) is known as the multiple measurement vector or MMV problem when $L>1$ \cite{cotter05}, compared to the single measurement vector or SMV problem when just a single snapshot is available i.e., $L=1$. The non-zero rows correspond to active sources, and one of the key problems in SSR is to identify these non-zero rows. For the gridless approach\cite{tang13,steffens18,yang16,yang15} the reparameterization involves \emph{Toeplitz} matrix fitting of appropriate size. Note that modern techniques are applicable more generally even when there is no underlying parametric model, for example Gaussian random entries in $\mathbf{\Phi}$. Sparsity can be explicitly enforced by adding suitable $p$-pseudo-mixed norm\footnote{Note that for $p=1$ we get a norm, as it satisfies all the required axioms.} ($p\in(0,1]$), $\lVert\bar{\mathbf{X}}\rVert_{2,p}$, regularizer for the grid case or atomic norm for the gridless formulations. The core emphasis in these approaches is on optimizing an appropriate fit to the measurements with an additional (sparsity) regularizer \cite{chen98,tropp06,malioutov05,tang13,steffens18,yang16,yang15}. Such methods are therefore sensitive to setting the regularization parameter properly. An exception to the explicit regularization based methods includes sparse Bayesian learning (SBL) \cite{tipping01,wipf04,wipfrao07} which recovers sparse solutions for (\ref{eq:sblparmodl}) via implicit regularization\cite{wipf10}. SBL formulates the recovery problem under the MLE framework and therefore demonstrates superior performance.

\noindent The question we seek to answer is: \emph{how can we enhance the SBL formulation to overcome the model complexities faced by MLE methods of the past, and solve (\ref{eq:parmodl}) i.e., perform gridless estimation of $\bm{\theta}$?}
We identify the following contributions:\begin{itemize}[leftmargin=*]\item It was shown in \cite{pal15} that correlation-aware techniques effectively utilize available geometry and prior information and thus, can recover support as high as $O(M^2)$. In \cite{balkan14}, it was shown that SBL can indeed identify $O(M^2)$ sources in the noiseless case under certain sufficient conditions on the dictionary and sources, and was shown empirically in the noisy case. In this work we reexamine the SBL formulation and show that it places a similar emphasis on available structure i.e., geometry and prior information, and thus is a correlation-aware technique! \item We reformulate the SBL problem as a novel \emph{structured matrix recovery} (SMR) problem \emph{under the MLE framework}.
We will also show that the cost function employed by the proposed method can be derived using the Kullback-Leibler (KL) divergence between the true (data) distribution and the one assumed in this work. This insight provides a new perspective for understanding the underlying strategy to handle the case when sources may be arbitrarily correlated, extending the benefits of correlation-aware methods.\item A majorization-minimization (MM) procedure\cite{sun17} to minimize the negative log-likelihood function is provided.
One of the advantages of such an approach over other algorithms like sequential quadratic programming (SQP) is that more information is retained as we only majorize the concave terms in the cost. Thus, all information about third order and higher, of the convex terms is retained, unlike in SQP. Also, unlike SQPs where trust regions are required which limit progress per iteration, such conservative measures are prevented using convex-concave procedure (CCP)\cite{lipp16}. Thus, the linear MM procedure allows for more progress per iteration. We further discuss how array geometry can play an important role in identifying more sources than sensors. 
We also provide perspectives to understand the proposed approach and connect with the traditional MLE framework and the modern SBL formulation.\item Finally, we consider arbitrary geometries where it is difficult to identify simplifying structures, that are otherwise possible for array geometries such as uniform linear arrays (ULA) with potential missing sensors. For this case, we propose adaptive grid-based strategies to extend SBL to alleviate the initial grid limitation.\end{itemize}The proposed techniques set us apart from other family of approaches in the literature that albeit put together a cost function with a similar essence (i.e. $\mathrm{Simple\>Model+Data\>Fitting}$), but lack a (MLE) principled approach and hence the associated insights, performance guarantees and rich options. We provide numerical results to further elucidate the impact of the proposed techniques
and compare them with other gridless approaches and the Cram\'er-Rao bound (CRB). Some of the work presented here was also discussed in \cite{pote22} by the authors. We will now review some relevant prior work in this field.\begin{table*}
    \centering
    \begin{tabular}{l|c|c|c|c|c}
    \hline
        {\bf Methods} & {\bf \begin{tabular}{c}Applicable Array\\Geometries\end{tabular}} & {\bf Sparsity Regularizer} & {\bf Iterative} & {\bf \begin{tabular}{c}Optimization\\Tool$^*$\end{tabular}} & {\bf \begin{tabular}{c}Knowledge of\\noise variance $(\sigma_n^2)$\end{tabular}}\\\hline
        ANM\cite{tang13} & $\mbox{ULA}^\dagger$ & Atomic norm & No & SDP & known\\
        RAM\cite{yang16} & $\mbox{ULA}^\dagger$ & Reweighted atomic norm & Yes & SDP & known\\
        Gridless SPARROW\cite{steffens18,wassim17} & $\mbox{ULA}^\dagger$ & $\ell_{2,1}$ mixed norm & No & SDP & known\\
        Gridless SPICE\cite{yang15} & $\mbox{ULA}^\dagger$ & Implicit (trace norm) & No & SDP & {\it unknown}\\
        Proposed & {\it Non-uniform linear array}
        & Implicit ($\log\det$) & Yes & SDP & known$^{**}$\\\hline
    \end{tabular}\vspace{0.2cm}\\\begin{tabular}{l}
    $\mbox{ULA}^\dagger$ includes ULA with missing sensors' case. Non-uniform linear array includes $\mbox{ULA}^\dagger$ as a special case. $^*$First-order methods have been proposed\\for some of the above algorithms, although they were primarily derived as SDPs. $^{**}\sigma_n^2$ can be assumed unknown and estimated as part of the procedure.\vspace{0.2cm}
    \end{tabular}
    \caption{Summary of Gridless Sparse Signal Recovery Algorithms\vspace{-0.8cm}}
    \label{tab:gridlessssr}
\end{table*}\subsection{Relevant Prior Work}
\noindent Early works, primarily in the field of DoA estimation using the MLE based cost function include\cite{burg82,bohme86,miller87, jaffer88,fuhrmann88,li99}.
In \cite{burg82}, the authors proposed an iterative algorithm to solve the necessary gradient equations for moderate sized problems. An expectation-maximization (EM) based approach was proposed in \cite{miller87} wherein the incomplete observed data is assumed to have a Toeplitz structured covariance, and where it is shown that it is possible to embed the incomplete data into a larger size periodic data series. 
A separable solution, consisting of an optimization problem for recovering support and a closed form expression for estimating the source covariance matrix was proposed in \cite{jaffer88}, which was further extended to the case when noise variance is unknown in \cite{stoica95}. The problem was later considered in the presence of spatially correlated noise fields in \cite{viberg97}. A closed-form formula for estimating Hermitian Toeplitz covariance matrices using the extended invariance principle was suggested in \cite{li99}. A covariance matching based estimation to bypass the model complexity associated with the MLE based cost function was proposed in \cite{ottersten98}.
The approach developed in this paper can be viewed as a natural progression of this line of work, benefiting from the developments in the field and in optimization tools.

\noindent In \cite{tang13}, authors proposed a gridless scheme for estimating the frequency components of a mixture of complex sinusoids based on the concept of atomic norm \cite{chandrasekaran12}. They formulated a semidefinite program (SDP) which recovered a low rank Toeplitz matrix. Such a Toeplitz matrix can be further decomposed to identify the DoAs. In our work we similarly break the task into two steps. First, we recover a structured covariance matrix approximation for the sample covariance matrix (SCM). This recovery is based on the MLE cost function, unlike the work in \cite{tang13}. The second step is similar to that in \cite{tang13}. At each step we process the SCM, and do not process the received samples directly. As a result, the problem dimension is bounded, and results into a \emph{compact} formulation. A similar compact reformulation, called SPARse ROW-norm reconstruction (SPARROW), for the atomic norm minimization problem was proposed in \cite{steffens18}. The atomic-norm minimization (ANM) technique in \cite{tang13} builds on the mathematical theory of super-resolution developed by Cand\'es et al.\cite{candes14}, in that it extends to the cases of partial/compressive samples and/or multiple measurement vectors. ANM, however, requires sources to be adequately separated, prohibiting true super-resolution. A re-weighted ANM (RAM) strategy that potentially overcomes the shortfalls of ANM was proposed in \cite{yang16}. SParse Iterative Covariance-based Estimation (SPICE) was proposed in \cite{stoica11} as a grid-based sparse parameter estimation technique based on covariance matching, as opposed to the MLE formulation, and was later extended to the gridless case in \cite{yang15}. It was shown in \cite{yang15} that gridless SPICE and atomic norm-based techniques are equivalent, under varied assumptions of noise. LIKelihood-based Estimation of Sparse parameters (LIKES) \cite{stoica12} was proposed as a grid-based method following the MLE principle, with the same application as SPICE.
Table~\ref{tab:gridlessssr} summarizes recent gridless SSR approaches.
\subsection{Organization of the Paper and Notations}
\noindent In Section~\ref{sec:bgdcaw} we begin with a simple insight into SBL formulation, and demonstrate that SBL is a correlation-aware technique. We further compare SBL with another line of correlation-aware algorithms based on minimizing diversity measures. We take this insight further and present the structured matrix recovery (SMR) reformulation and highlight benefits of the proposed approach when sources may be arbitrarily correlated. In Section~\ref{sec:structcovmle}, we propose an iterative algorithm to solve the SMR problem. We consider both ULA without missing sensors and ULA wherein some sensors may be missing, in this section. We also connect the proposed SMR approach with the traditional MLE framework and the modern SBL formulation. In Section~\ref{sec:sblgdrefine}, we discuss the general case where sensors may be placed arbitrarily, and may not lie on a uniform grid. We present numerical results in Section~\ref{sec:sim} and conclude the work in Section~\ref{sec:conc}.

\noindent We represent scalars, vectors, and matrices by lowercase, boldface-lowercase, and boldface-uppercase letters, respectively. Sets are represented using blackboard bold letters. $(.)^T$ denotes transpose and $(.)^H$ denotes Hermitian of the operand matrix, and $(.)^c$ denotes element-wise complex conjugate.
$\odot$ denotes Khatri-Rao product between two matrices of appropriate sizes.
\section{SBL Revisited: Correlation Aware Interpretation, Robustness, and Structured Matrix Reformulation
}\label{sec:bgdcaw}
\noindent A correlation-aware technique\cite{pal12,pal15} satisfies the following three general requirements\footnote{To our knowledge, a formal set of requirements to be a `correlation-aware' technique is missing in literature. Thus, we propose these requirements based on the conditions for superior source identifiability reported in \cite{pal12,pal15,balkan14}.}: a) it depends on the measurements only through its second order statistics b) it assumes a
source correlation prior, usually that sources are uncorrelated, and fits a resulting \emph{structured} received signal covariance matrix to the second order statistics of the measurements c) any further inference is carried using the recovered parameters characterizing the estimated structured covariance matrix. In this work we assume that the sources are uncorrelated. This assumption may not always hold, and some sources may in fact be correlated. The impact of this mismatch between assumed model and true model is discussed at the end of this section.
The discussion highlights another aspect of the MLE framework, as it provides interpretable and superior results even in the mismatched model case.

\noindent For the purpose of simplicity, we focus on the ULA geometry in this section, and postpone the general case of ULAs with \emph{missing sensors} until next section. However, the insights presented here are applicable to the general case as well.
\subsection{On the SBL Algorithm}
\noindent SBL is a Bayesian technique to find a row-sparse decomposition of the received measurements, $\mathbf{Y}=[\mathbf{y}_0,\ldots,\mathbf{y}_{L-1}]$, (i.e., to solve the MMV problem in (\ref{eq:sblparmodl})) using an overcomplete dictionary $\mathbf{\Phi}\in\mathbb{C}^{M\times G}$ consisting of $G$ suitably chosen vectors (may be \emph{non-parametric} in general). In the DoA estimation problem, these vectors are array manifold vectors evaluated on a grid of angular space representing potential DoAs i.e., $\theta\in[-\frac{\pi}{2},\frac{\pi}{2})$ or $u\in[-1,1)$ in $u$-space. 
Note that there is a bijective mapping  $u=\sin\theta$ in the domains of interest\cite{trees02} and thus we use the two notations interchangeably. Consider a ULA with $M$ sensors and $d=\bar{\lambda}/2$ distance between adjacent sensors to prevent ambiguity in DoA estimation; $\bar{\lambda}$ denotes the wavelength of the incoming narrowband source signals. The array manifold vector for a source signal incoming at angle $u\in[-1,1)$, is given by $\bm{\phi}(u)=\left[1, \exp{(-j\pi u)},\ldots,\exp{(-j(M-1)\pi u)}\right]^T$.

\noindent SBL imposes a \emph{parameterized} Gaussian prior on the source signal $\bar{\mathbf{x}}_l\in\mathbb{C}^G$ as $\bar{\mathbf{x}}_l\sim\mathcal{CN}(\mathbf{0},\bm{\Gamma})$. Note that SBL explicitly imposes an \emph{uncorrelated sources} prior, and thus $\bm{\Gamma}$ is a diagonal matrix; let $\mathrm{diag}(\bm{\Gamma})=\bm{\gamma}$.
Thus we have $\mathbf{y}_l\sim\mathcal{CN}(\mathbf{0},\mathbf{\Phi}\bm{\Gamma}\mathbf{\Phi}^H+\lambda\mathbf{I})$, $\lambda$ denotes the estimate for noise variance. In the case with uninformative prior for $\bm{\gamma}$, the hyperparameter $\bm{\Gamma}$ and $\lambda$ can be estimated under the MLE framework\cite{wipfrao07} as\begin{equation}
    \underset{\mathbf{\Gamma}\succeq\mathbf{0},\>\lambda\geq 0}{\min}\log\det\left(\mathbf{\Phi\Gamma\Phi}^H+\lambda\mathbf{I}\right)+\mathrm{tr}\left(\left(\mathbf{\Phi\Gamma\Phi}^H+\lambda\mathbf{I}\right)^{-1}\hat{\mathbf{R}}_{\mathbf{y}}\right),\label{eq:sblopt}
\end{equation}where $\hat{\mathbf{R}}_{\mathbf{y}}=\frac{1}{L}\sum_{l=0}^{L-1}\mathbf{y}_l\mathbf{y}_l^H$ denotes the SCM.
Choices for solving the problem in (\ref{eq:sblopt}) include the Tipping iterations\cite{tipping01}, EM iterations\cite{wipf04}, sequential SBL\cite{tipping03}, and generalized approximate message passing (GAMP) implementations\cite{alshoukairi18,shengheng19}. A MM approach for solving (\ref{eq:sblopt}) was introduced in \cite{wipf08}. 
\begin{remark}
Note that if the number of sources $K$ is known exactly in (\ref{eq:parmodl}), such model order information is not used in the SBL formulation. Instead, the $\log\det$ penalty in (\ref{eq:sblopt}) helps to promote sparsity and to deal with small but unknown number of sources. If there is prior knowledge on $K$, then $\lVert\bm{\gamma}\rVert_0 = K$ would have to be imposed on the objective function.\end{remark}\noindent We now present the following useful insight.\begin{proposition}$\forall\bm{\gamma}\geq\mathbf{0}$ such that $(\mathbf{\Phi}\odot\mathbf{\Phi}^c)\bm{\gamma}=\mathbf{w}$, for some fixed $\mathbf{w}\in\mathbb{C}^{M^2}$, the SBL cost is a constant i.e.,\begin{equation*}
    \log\det\left(\mathbf{\Phi\Gamma\Phi}^H+\lambda\mathbf{I}\right)+\mathrm{tr}\left(\left(\mathbf{\Phi\Gamma\Phi}^H+\lambda\mathbf{I}\right)^{-1}\hat{\mathbf{R}}_{\mathbf{y}}\right)=C(\lambda),
\end{equation*}where $C(\lambda)$ is some constant.\vspace{-0.2cm}
\end{proposition}\begin{proof}
The proof follows simply by observing that $(\mathbf{\Phi}\odot\mathbf{\Phi}^c)\bm{\gamma}=\mathbf{w}$ implies $\mathbf{\Phi\Gamma\Phi}^H$ is a fixed \emph{structured} matrix with entries dictated by components of $\mathbf{w}$.\vspace{-0.2cm}
\end{proof}\noindent The above result demonstrates that, the hyperparameter $\bm{\gamma}$ affects the SBL cost function only through the entries of the \emph{structured covariance matrix} of the measurements. The sources are localized by peaks in the output $\bm{\gamma}$ pseudospectrum. This procedure satisfies the general requirements for correlation-aware algorithms. Thus, we conclude that SBL is indeed a correlation-aware technique. The procedure also marks some key requirements for superior sources' identifiability (see Theorem 1 and following remarks in \cite{balkan14}).
\subsection{Connecting to Correlation-Aware SSR Techniques based on Minimizing Diversity Measures}
\noindent Consider the class of problems given by\begin{IEEEeqnarray}{lll}
\min_{\mathbf{z}\geq\mathbf{0}}&f(\mathbf{z})&\label{eq:qiaocorraw}\\
\mathrm{subject\>to}\quad&\lVert\hat{\mathbf{r}}_{\mathbf{y}}-\mathbf{\Phi}_{KR}\mathbf{z}\rVert_2\leq\epsilon,&\nonumber
\end{IEEEeqnarray}where, $\hat{\mathbf{r}}_{\mathbf{y}}=\mathrm{vec}(\hat{\mathbf{R}}_{\mathbf{y}})$ and $\mathbf{\Phi}_{KR}=\mathbf{\Phi}^c\odot\mathbf{\Phi}$ denotes the Khatri-Rao product of $\mathbf{\Phi}$ with its conjugate. $f(\mathbf{z})$ is a sparsity promoting objective function and choices include $\ell_1$ norm, $\ell_0$ or $\ell_{1/2}$ as considered in \cite{qiao19}. The above problem satisfies the requirements for being correlation-aware, namely a) it matches the model to the second order statistics of the data b) uses uncorrelated sources' correlation prior to fit a structured matrix to the measurements c) further performs inference using the parameters of this estimated structured matrix. Next, we reformulate SBL as a constrained optimization problem to highlight the data-fitting term and to compare with (\ref{eq:qiaocorraw}).

\noindent The MLE optimization problem in (\ref{eq:sblopt}) can be reformulated as a constrained optimization problem as follows:\begin{IEEEeqnarray}{lcl}
\underset{\mathbf{\Gamma}\succeq\mathbf{0},\lambda\geq 0}{\min}&\quad&\log\det\left(\mathbf{\Phi\Gamma\Phi}^H+\lambda\mathbf{I}\right)\label{eq:toepcon}\\
\mathrm{subject\>to}&&\mathrm{tr}\left(\left(\mathbf{\Phi\Gamma\Phi}^H+\lambda\mathbf{I}\right)^{-1}\hat{\mathbf{R}}_{\mathbf{y}}\right)\leq\epsilon\nonumber.
\end{IEEEeqnarray}Note that the constraint imposes a Mahalanobis distance-based bound on the optimization variables. Another perspective to understand the data-fitting term above based on regularized least-squares fit to measurements can be found in \cite{wipf08}.
\begin{proposition}
Let $(\mathbf{\Gamma}^*,\lambda^{*})$ be a global minimizer of the optimization problem in (\ref{eq:sblopt}) such that $\lambda^{*}>0$. $(\mathbf{\Gamma}^*,\lambda^{*})$ globally minimizes problem in (\ref{eq:toepcon}) as well, if and only if $\epsilon=\mathrm{tr}\left(\left(\mathbf{\Phi\Gamma}^*\mathbf{\Phi}^H+\lambda^{*}\mathbf{I}\right)^{-1}\hat{\mathbf{R}}_{\mathbf{y}}\right)$.\label{prop:toepcon}
\end{proposition}
\begin{proof}
    Proof in Appendix section~\ref{sec:app_prop}.
\end{proof}
\noindent The constraint in the formulation of (\ref{eq:toepcon}) allows to match the model to the observation (through the sample covariance matrix) and the objective function promotes a \emph{simpler} model to be picked. Note that the constrained optimization problem in (\ref{eq:toepcon}) is exactly MLE only when $\epsilon$ is set appropriately. The proposition indicates the difficulty in transforming the MLE to a constrained problem, although the latter can be explored as a viable option with $\epsilon$ set heuristically. This is not discussed further and left as future work. The above outlook only tries to highlight the two components of the SBL objective and allows one to compare the constrained formulation in (\ref{eq:toepcon}) to the other \emph{correlation-aware} technique in (\ref{eq:qiaocorraw}). The data fitting term in (\ref{eq:qiaocorraw}) lacks the MLE framework for data fitting used in (\ref{eq:toepcon}). This insight highlights one of the key difference between our approach and that used in many other works in the literature.

\noindent An alternative treatment of the SBL cost function that also reveals connections to reweighted $\ell_1$ and $\ell_2$ methods for finding sparse solutions to (\ref{eq:sblparmodl}) can be found in \cite{wipf08, wipf11, wipf10}.
\subsection{Proposed SMR Approach: ULA with No Missing Sensors}
\noindent The structure for $\mathbf{\Phi\Gamma\Phi}^H$ in the case of ULA is a Toeplitz matrix, and is informed by the array geometry and the uncorrelated sources prior. In other words, SBL attempts to find the `best' positive semidefinite (PSD) Toeplitz matrix approximation to the SCM $\hat{\mathbf{R}}_{\mathbf{y}}$. The grid-based formulation restricts the solution to lie in the union of PSD cones. We use this insight and reparameterize the SBL cost function to directly estimate the entries of the Toeplitz covariance matrix. Let $\mathbf{v}$ denote the first row of such a Toeplitz matrix, denoted by $\T(\mathbf{v})$. We reformulate the SBL optimization problem as\begin{equation}
\underset{\substack{\mathbf{v}\in\mathbb{C}^{M}\text{ s.t. }\\\T(\mathbf{v})\succeq\mathbf{0}, \lambda\geq 0}}{\min}\log\det\left(\T(\mathbf{v})+\lambda\mathbf{I}\right)+\mathrm{tr}\left((\T(\mathbf{v})+\lambda\mathbf{I})^{-1}\hat{\mathbf{R}}_{\mathbf{y}}\right).\label{eq:Toepcost}\end{equation}Once the solution $\mathbf{v}^*$ is obtained, we estimate the DoAs by decomposing the Toeplitz matrix, $\T(\mathbf{v^*})$. In our simulations we use root-MUSIC to estimate the DoAs \cite{barabell83}.
\begin{remark}It is known that a low rank ($D<M$) PSD Toeplitz matrix such as $\T(\mathbf{v}^*)$ can be uniquely decomposed as $\T(\mathbf{v}^*)=\sum_{i=1}^Dp_i\bm{\phi}(\theta_i)\bm{\phi}(\theta_i)^H,p_i>0$, and $\theta_i$'s are distinct
\cite{caratheodoryfejer11}. 
In (\ref{eq:Toepcost}), a low-rank solution is encouraged by the $\log\det$ term \cite{fazel03}, while its effect is being moderated by the additional noise variance term, `$+\lambda\mathbf{I}$'.\end{remark}\noindent The SBL formulation in (\ref{eq:sblopt}) not only finds a structured matrix fit to the measurements, it also \emph{factorizes} it. The same is true with the classical MLE approach, and is briefly discussed in Section~\ref{sec:reparamdiscuss}. The structured matrix factorization is a crucial step. In the proposed approach, we find a structured matrix in the MLE sense.
We therefore refer to the proposed approach as `StructCovMLE'. The problem in (\ref{eq:Toepcost}) is non-convex and we discuss an iterative algorithm to solve it, along with an extension to allow ULAs with missing sensors, in Section~\ref{sec:structcovmle}.

\noindent Next, we briefly discuss an important aspect of the chosen approach in (\ref{eq:Toepcost}) to solve the original problem in (\ref{eq:parmodl}).
\subsection{Performance under a Correlation Prior Mismatch}
\noindent We discuss the case when there is a prior misfit, between the assumed model and the actual (data) model. This insight is another feature resulting from the MLE formulation used by SBL as opposed to a regularization framework. In particular, we discuss the case when the sources may be arbitrarily correlated. As briefly mentioned before, in the case of a ULA, the structure SBL imposes by virtue of the array geometry and the (uncorrelated) source correlation prior is a Toeplitz matrix. If some of the sources are correlated, the approach fits Toeplitz structured covariance to a non-Toeplitz structure obeyed by the data. Our aim is not to correct but to quantify the model misfit. In particular, we show that the recovered Toeplitz fit to the SCM minimizes the KL divergence between the assumed and the true distribution.
 
\noindent Let $p_\mathbf{y}$ and $f_{\mathbf{y}\mid\bm{\Psi}}$ denote the true probability density function (pdf) and the pdf for the mismatched model, respectively, where $\bm{\Psi}=(\mathbf{v},\lambda)$ s.t. $\T(\mathbf{v})\succeq\mathbf{0},\lambda\geq 0$. Since the source and noise vectors are uncorrelated with each other, $p_\mathbf{y}$ is a zero mean Gaussian pdf with covariance matrix $\mathbf{R_y}=\mathbf{\Phi}_{\bm{\theta}}\mathbf{R_x}\mathbf{\Phi}_{\bm{\theta}}^H+\sigma_n^2\mathbf{I}$, where $\mathbf{R_x}$ denotes the source covariance matrix. Similarly $f_{\mathbf{y}\mid\bm{\Psi}}$ is zero mean Gaussian pdf with covariance $\mathbf{\Sigma_y}=\T(\mathbf{v})+\lambda\mathbf{I}$. The KL divergence between these two normal distributions is well known and is given by\begin{IEEEeqnarray}{rcl}
	D(p_\mathbf{y}\Vert f_{\mathbf{y}\mid\bm{\Psi}})&=\>&\log\det\mathbf{\Sigma_y}-\log\det\mathbf{R_y}-M+\mathrm{tr}(\mathbf{\Sigma_y}^{-1}\mathbf{R_y}).
\end{IEEEeqnarray}The effective optimization problem to minimize the KL divergence between the two distributions is given by\begin{equation}
    \mathbf{\Psi}^*=\underset{\mathbf{\Psi} \mbox{ s.t. }\T(\mathbf{v})\succeq\mathbf{0},\lambda\geq 0}{\text{argmin}}\>\log\det\bigl(\mathbf{\Sigma_y}\bigr)+\mathrm{tr}\bigl(\mathbf{\Sigma_y}^{-1}\mathbf{R_y}\bigr).\label{eq:ObFnmain}
\end{equation}Note that this optimization problem is similar to (\ref{eq:Toepcost}) for the proposed approach (or (\ref{eq:sblopt}) used within SBL), where instead of the actual received signal covariance matrix, $\mathbf{R_y}$, we used the SCM. Note that the SCM is the unconstrained/unstructured MLE estimate of the received signal covariance matrix. In \cite{pote20} it was shown for SBL using the two sources example and when the DoAs were known that, when sources are far apart, the estimate for the source powers under the uncorrelated model matches the true source power using the problem in (\ref{eq:ObFnmain}). Such a mismatched model was also used in \cite{wipf07} to propose more robust beamformers that can resist source correlation.

\section{Maximum Likelihood Structured Covariance Matrix Recovery}\label{sec:structcovmle}
\noindent We focus on ULA, first on the case with no missing sensors, and then on the case of ULA with missing sensors. We assume that the noise variance is known and set $\lambda=\sigma_n^2$ in (\ref{eq:Toepcost}), but it can be estimated as well, similar to $\mathbf{v}$ in this section.
\subsection{Uniform Linear Array Geometry}
\noindent Based on the concavity of the $\log\det$ term, we  majorize the $\log\det$ term in (\ref{eq:Toepcost}) and replace it with a linear term using its Taylor expansion\cite{sun17}\begin{IEEEeqnarray}{ll}
    \log\det\left(\T(\mathbf{v})\right.&\left.+\lambda\mathbf{I}\right)\leq\log\det\left(\T(\mathbf{v}^{(k)})+\lambda\mathbf{I}\right)\nonumber\\&+\mathrm{tr}\left((\T(\mathbf{v}^{(k)})+\lambda\mathbf{I})^{-1}\T(\mathbf{v}-\mathbf{v}^{(k)})\right),\IEEEeqnarraynumspace
\end{IEEEeqnarray}where $\mathbf{v}^{(k)}$ denotes the iterate value at the $k$th iteration. 
Note that the linear term from Taylor expansion provides a supporting hyperplane to the hypograph $\{(\mathbf{v},t): t<=\log\det(\T(\mathbf{v})+\lambda \mathbf{I}))\}$\cite{boyd04}. We ignore the constant terms above and get the following majorized objective function
\begin{equation*}
 \mathrm{tr}\left((\T(\mathbf{v}^{(k)})+\lambda\mathbf{I})^{-1}\T(\mathbf{v})\right)+\mathrm{tr}\left((\T(\mathbf{v})+\lambda\mathbf{I})^{-1}\hat{\mathbf{R}}_{\mathbf{y}}\right).
 \end{equation*}
 Rewriting second term above using Schur complement lemma:
 \begin{IEEEeqnarray}{lcl}
 \mathrm{tr}\left((\T(\mathbf{v})+\lambda\mathbf{I})^{-1}\hat{\mathbf{R}}_{\mathbf{y}}\right)\>&=&\underset{\mathbf{U}\in\mathbb{C}^{M\times M}}{\min}\mathrm{tr}\left(\mathbf{U}\>\hat{\mathbf{R}}_{\mathbf{y}}\right)\nonumber\\&&\mathrm{s.t.}\left[\begin{array}{cc}
    \mathbf{U} & \mathbf{I}_M \\
    \mathbf{I}_M & \T(\mathbf{v})+\lambda\mathbf{I} 
\end{array}\right]\succeq\mathbf{0},\qquad
\end{IEEEeqnarray}
which is a SDP. The overall optimization problem is convex and can be formulated as a SDP as follows\begin{IEEEeqnarray}{lcl}
\underset{\mathbf{v}\in\mathbb{C}^{M},\mathbf{U}\in\mathbb{C}^{M\times M}}{\min}&&\>\mathrm{tr}\left((\T(\mathbf{v}^{(k)})+\lambda\mathbf{I})^{-1}\T(\mathbf{v})\right)+\mathrm{tr}\left(\mathbf{U}\>\hat{\mathbf{R}}_{\mathbf{y}}\right)\nonumber\\\mathrm{subject\>to}&&\left[\begin{array}{cc}
    \mathbf{U} & \mathbf{I}_M \\
    \mathbf{I}_M & \T(\mathbf{v})+\lambda\mathbf{I} 
\end{array}\right]\succeq\mathbf{0},\T(\mathbf{v})\succeq\mathbf{0},\label{eq:sdptoep}
\end{IEEEeqnarray}and can be solved using any standard solvers (e.g. CVX solvers such as SDPT3, SeDuMi\cite{cvx}). It can be solved iteratively and we summarize the proposed steps in Algorithm~\ref{alg:glssr}. The following remark briefly discusses the choice of initialization.\begin{algorithm}
	\SetAlgoLined
	\KwResult{$\mathbf{v}^*$}
	\KwIn{Measurements: $\mathbf{Y}\in\mathbb{C}^{M\times L},\lambda=\sigma_n^2,\mathrm{ITER}$}
	Initialize: $\hat{\mathbf{R}}_{\mathbf{y}}=\mathbf{YY}^H/L, \mathbf{v}^*=\mathbf{e}_1=[1,0,\ldots,0]^T$\\
	\For{$k:=1\>\mathrm{to\>\mathrm{ITER}}$}{
	$\mathbf{v}^{(k)}\leftarrow\mathbf{v}^*$\\
	$\mathbf{v}^*\leftarrow\mbox{Solve the problem in}$~(\ref{eq:sdptoep})}
	\caption{Proposed `StructCovMLE' Algorithm}\label{alg:glssr}
\end{algorithm}\begin{remark}\label{rem:initial}We initialize the proposed algorithm with the unit vector $\mathbf{v}_0=\mathbf{e}_1$, following the suggestion in \cite{fazel03} for effective rank minimization. 
This initialization reduces the majorized term to a trace function in the first iteration. It is known that trace function is a convex envelope for the rank function for matrices with spectral norm less than one\cite{fazel01}. Furthermore, the iterative weighted trace minimization in the following iterations helps to \emph{preserve relevant signal components}.\end{remark}
\subsection{ULA with Missing Sensors}\noindent We begin by identifying the relevant \emph{structure} for the general case of ULAs with missing sensors. Consider a linear array with $M$ sensors on a grid with minimum inter-element spacing $d=\bar{\lambda}/2$. Let $\mathbb{P}=\{p_i\mid p_i\in\mathbb{Z}, 0\leq i< M\}$ denote the set of normalized (w.r.t. $d$) sensor positions. We assume $p_0=0$ without loss of generality. The array manifold vector is given by $\bm{\phi}(u)=\left[1, \exp{(-jp_1\pi u)},\ldots,\exp{(-jp_{M-1}\pi u)}\right]^T,u=\sin\theta$. The difference coarray is given by $\mathbb{D} = \{z\mid z=r-s,\>r,s\in\mathbb{P}\}.$ The concept of difference coarray influences the structure we seek to identify, and also arises naturally when computing the received signal covariance matrix. It represents the set of unique lags experienced by the physical array. 
The received signal covariance matrix under the SBL formulation is given by $\mathbf{\Phi\Gamma\Phi}^H+\lambda\mathbf{I}$, as discussed previously. The $(m,n)$ entry in $\mathbf{\Phi\Gamma\Phi}^H$ is given by $[\mathbf{\Phi\Gamma\Phi}^H]_{m,n}=\sum_{i=1}^G\>\gamma_i\exp{(-j(p_m-p_n)\pi u_i)}$, and $[\mathbf{\Phi\Gamma\Phi}^H]_{m,n}=[\mathbf{\Phi\Gamma\Phi}^H]_{n,m}^c$. Thus, $[\mathbf{\Phi\Gamma\Phi}^H]_{m,n}=[\mathbf{\Phi\Gamma\Phi}^H]_{m',n'},\forall$ tuples $(m,n)$ and $(m',n')$ such that $p_m-p_n=p_{m'}-p_{n'}$. In other words, the entries in $\mathbf{\Phi\Gamma\Phi}^H$ can be distinct only corresponding to distinct elements in $\mathbb{D}$. $\mathbf{\Phi\Gamma\Phi}^H$ is Hermitian symmetric, which further restricts the number of distinct entries. This reveals the underlying \emph{structure} that the model $\mathbf{\Phi\Gamma\Phi}^H$ satisfies, and we formalize it below.

\noindent Let $M_{\mathrm{apt}}$ denote the aperture of the array, $M_{\mathrm{apt}}=\max_{d\in\mathbb{D}}d+1\label{eq:aptdefn}$. We define a linear mapping $\mathbf{T}(\mathbf{v}):\mathbb{C}^{M_{\mathrm{apt}}}\rightarrow\mathbb{C}^{M\times M}$ as\begin{IEEEeqnarray}{rlr}
    [\mathbf{T}(\mathbf{v})]_{i,j}&=\left\{\begin{array}{ll}
        v_{\vert p_i-p_j\vert} &  j\geq i\\
        v_{\vert p_i-p_j\vert}^c & \text{otherwise}
    \end{array}\right.,&\> 0\leq i,j< M.\IEEEeqnarraynumspace
\label{eq:structmatdef}\end{IEEEeqnarray}
The mapping $\mathbf{T}(\mathbf{v})$ in general is many-to-one. It is only when the difference coarray has no holes, the mapping is one-to-one. For such cases we define $\mathbf{T}^{-1}(\mathbf{R}):\mathbb{C}^{M\times M}\rightarrow\mathbb{C}^{M_{\mathrm{apt}}}$ as a function that extracts the entries of a given structured matrix $\mathbf{R}$, formed using (\ref{eq:structmatdef}), to form a column vector. For the ULA with no missing sensors' case, we have $\mathbf{T}(\mathbf{v})=\T(\mathbf{v})$.\begin{figure}[h]
    \centering
    \begin{tabular}{cc}
    \includegraphics[width=0.4\linewidth]{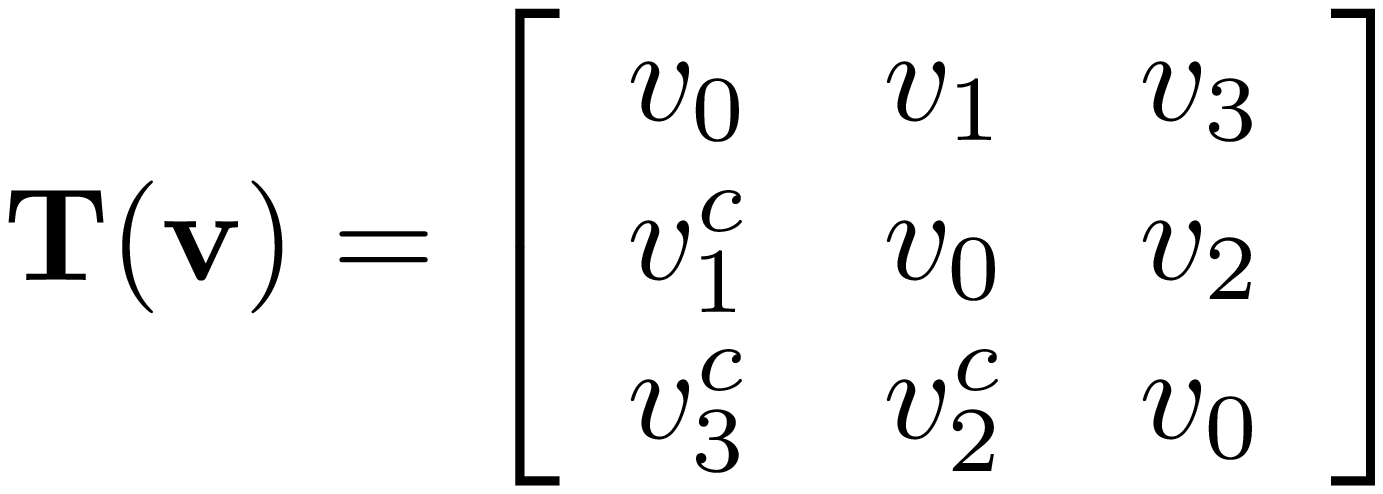} & \includegraphics[width=0.4\linewidth]{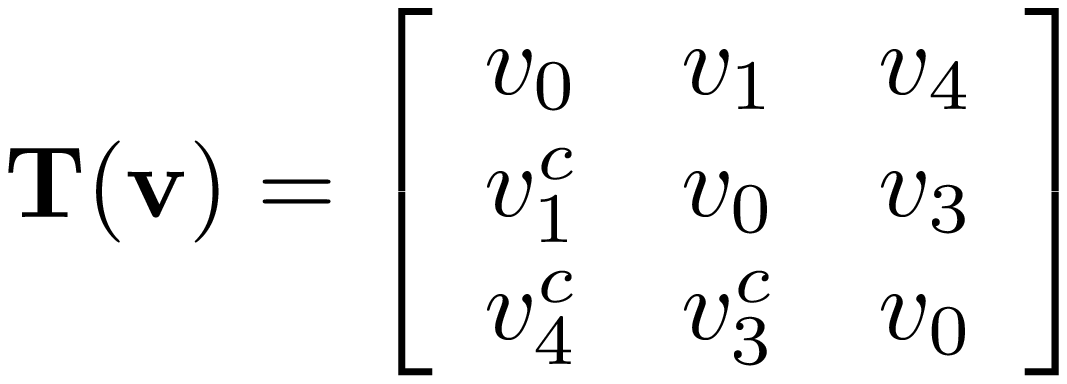}\\
    (a) & (b)\end{tabular}
    \caption{Structured Covariance Matrix $\mathbf{T}(\mathbf{v})$}
    \label{fig:Egstructcovmat}
\end{figure}\\\noindent{\bf Example 1:} Consider $\mathbb{P}=\{0,1,3\}$. This leads to $\mathbb{D}=\{-3,-2,-1,0,1,2,3\}$ and $M_{\mathrm{apt}}=4$. We therefore define $\mathbf{v}\in\mathbb{C}^4$, and the structured coavariance matrix as in Fig.~\ref{fig:Egstructcovmat}(a).
The mapping $\mathbf{T}(\mathbf{v})$ is one-to-one here and consequently $\mathbf{T}^{-1}$ is defined and we have $\mathbf{T}^{-1}(\mathbf{T}(\mathbf{v}))=[ v_0,v_1,v_2,v_3]^T$.\\
\noindent{\bf Example 2:} Consider $\mathbb{P}=\{0,1,4\}$. This leads to $\mathbb{D}=\{-4,-3,-1,0,1,3,4\}$ and $M_{\mathrm{apt}}=5$. We therefore define $\mathbf{v}\in\mathbb{C}^5$, and the structured coavariance matrix as in Fig.~\ref{fig:Egstructcovmat}(b).
The mapping $\mathbf{T}(\mathbf{v})$ is many-to-one here, as the component $v_2$ is missing in $\mathbf{T}(\mathbf{v})$. Consequently $\mathbf{T}^{-1}$ is not defined.

\noindent Thus, for the general case, (\ref{eq:sblopt}) can be reformulated as:
\begin{equation}
\underset{\substack{\mathbf{v}\in\mathbb{C}^{M_{\mathrm{apt}}}\text{ s.t. }\\\T(\mathbf{v})\succeq\mathbf{0},\lambda\geq 0}}{\min}\>\log\det\left(\mathbf{T}(\mathbf{v})+\lambda\mathbf{I}\right)+\mathrm{tr}\left((\mathbf{T}(\mathbf{v})+\lambda\mathbf{I})^{-1}\hat{\mathbf{R}}_{\mathbf{y}}\right).\label{eq:mlcost3}    
\end{equation}\begin{remark}
We would like to highlight a non-trivial choice made above of imposing $\T(\mathbf{v})\succeq\mathbf{0}$, instead of only requiring $\mathbf{T}(\mathbf{v})\succeq\mathbf{0}$. Note that the former constraint ensures that the latter is satisfied. The choice imposes a relevant constraint and is an important aspect of the model we wish to fit to the data in MLE sense. It also helps to connect the proposed reformulation to the traditional and modern MLE approaches, and is discussed 
in Section~\ref{sec:reparamdiscuss}.
\end{remark}\begin{remark}As in the case for SBL, if the number of sources, $K$, is known, a rank constraint $\mathrm{rank}(\T(\mathbf{v}))=K$ should be imposed. Since imposing a rank constraint is difficult, 
surrogate measures like in compressed sensing may be used, such as `$+\beta\log\det(\T(\mathbf{v})+\epsilon\mathbf{I})$' as a regularizer in
(\ref{eq:mlcost3}) to further promote sparse solutions. In this work, we do not exploit knowledge of $K$ to solve
(\ref{eq:mlcost3}).\end{remark}\noindent Like in the previous case of ULA with no missing sensors, we majorize the cost function in (\ref{eq:mlcost3}) to get a convex function and rewrite it as a SDP, assuming knowledge of noise variance and setting $\lambda=\sigma_n^2$. The majorized objective is given by\begin{equation}
\mathrm{tr}\left((\mathbf{T}(\mathbf{v}^{(k)})+\lambda\mathbf{I})^{-1}\mathbf{T}(\mathbf{v})\right)+\mathrm{tr}\left((\mathbf{T}(\mathbf{v})+\lambda\mathbf{I})^{-1}\hat{\mathbf{R}}_{\mathbf{y}}\right).\label{eq:nulacost}
\end{equation}The resulting SDP is given below\begin{IEEEeqnarray}{lcl}
\underset{\mathbf{v}\in\mathbb{C}^{M_{\mathrm{apt}}},\mathbf{U}\in\mathbb{C}^{M\times M}}{\min}&&\>\mathrm{tr}\left((\mathbf{T}(\mathbf{v}^{(k)})+\lambda\mathbf{I})^{-1}\mathbf{T}(\mathbf{v})\right)+\mathrm{tr}\left(\mathbf{U}\>\hat{\mathbf{R}}_{\mathbf{y}}\right)\label{eq:genarr}\IEEEeqnarraynumspace\\\mathrm{subject\>to}&&\left[\begin{array}{cc}
    \mathbf{U} & \mathbf{I}_M \\
    \mathbf{I}_M & \mathbf{T}(\mathbf{v})+\lambda\mathbf{I} 
\end{array}\right]\succeq\mathbf{0},\T(\mathbf{v})\succeq\mathbf{0}, \nonumber   
\end{IEEEeqnarray}where $\mathbf{v}^{(k)}$ denotes the value at the $k$th iteration. Steps similar to Algorithm~\ref{alg:glssr} can be followed to find the optimal point $\mathbf{v}^*$. To estimate the DoAs we perform root-MUSIC on $\mathbf{T}(\mathbf{v}^*)$.\begin{remark}
It was shown in \cite{pillai85} that sparse arrays with a larger number of consecutive lags than the number of sensors, $M$, can identify more sources than $M$. Under the proposed approach, a similar higher identifiability can be achieved by instead performing root-MUSIC on $\T(\mathbf{v}^*)$, and we numerically verify this in Section~\ref{sec:moresources}.\vspace{-0.1cm}\end{remark}
\subsection{On Proposed Method: From MLE to SBL}\label{sec:reparamdiscuss}
\noindent We connect the proposed technique with the classical MLE framework and the grid SBL formulation. We hope to answer the following question: {\it how has the reparameterization affected the original problem in (\ref{eq:parmodl}) of solving for $\bm{\theta}$?}\\{\it 1) Connection with the classical MLE formulation:} 
We begin by first stating the traditional MLE formulation. In this approach, we impose a \emph{parametrized} Gaussian prior on $\mathbf{x}_l$ i.e., $\mathbf{x}_l\sim\mathcal{CN}(\mathbf{0},\mathbf{P})$. Note that an explicit knowledge of model order information is a requisite here. We further assume that the sources are uncorrelated, and thus $\mathbf{P}$ is a diagonal matrix. The resulting optimization problem is given by\begin{IEEEeqnarray}{ll}
\underset{\bm{\theta}\in[-\frac{\pi}{2},\frac{\pi}{2})^K,\mathbf{P}\succ\>\mathbf{0},\lambda\geq0}{\min}\>\log\det&\left(\mathbf{\Phi}_{\bm{\theta}}\mathbf{P\Phi}_{\bm{\theta}}^H+\lambda\mathbf{I}\right)\nonumber\\&+\mathrm{tr}\left((\mathbf{\Phi}_{\bm{\theta}} \mathbf{P\Phi}_{\bm{\theta}}^H+\lambda\mathbf{I})^{-1}\hat{\mathbf{R}}_{\mathbf{y}}\right).\IEEEeqnarraynumspace\label{eq:mlcost}    
\end{IEEEeqnarray}
The model is also referred to as the \emph{unconditional model} in the DoA literature\cite{stoica90}, compared to the \emph{conditional model} where $\mathbf{x}_l$ is assumed deterministic.
Consider the following updated MLE optimization problem:\begin{IEEEeqnarray}{ll}
\underset{\substack{K\in\mathbb{Z}^+\\0<K<M_{\mathrm{apt}}}}{\min}\>\underset{\bm{\theta}\in[-\frac{\pi}{2},\frac{\pi}{2})^K,\mathbf{P}\succ\>\mathbf{0},\lambda\geq0}{\min}\>&\log\det\left(\mathbf{\Phi}_{\bm{\theta}}\mathbf{P\Phi}_{\bm{\theta}}^H+\lambda\mathbf{I}\right)\nonumber\\&+\mathrm{tr}\left((\mathbf{\Phi}_{\bm{\theta}} \mathbf{P\Phi}_{\bm{\theta}}^H+\lambda\mathbf{I})^{-1}\hat{\mathbf{R}}_{\mathbf{y}}\right).\IEEEeqnarraynumspace\label{eq:mlcostKopt}
\end{IEEEeqnarray}The difference with the traditional MLE formulation is that, in the above we consider all model orders, $0<K<M_{\mathrm{apt}}$, to optimize the cost function. We then have the following result.
\begin{theorem}
The problem in (\ref{eq:mlcost3}) and in (\ref{eq:mlcostKopt}) are equivalent, in that they achieve the same globally minimum cost.\label{thm:mleeqprop}\end{theorem}
\begin{proof}Proof is provided in the Appendix section~\ref{prf:mleeqprop}.\end{proof}
\noindent{\it 2) Connection with SBL in (\ref{eq:sblopt}):} Consider the following updated SBL optimization problem:
\begin{equation}
    \underset{\mathbf{\Phi}}{\min}\>\>\underset{\mathbf{\Gamma}\succeq\mathbf{0},\lambda\geq 0}{\min}\>\>\log\det\left(\mathbf{\Phi\Gamma\Phi}^H+\lambda\mathbf{I}\right)+\mathrm{tr}\left(\left(\mathbf{\Phi\Gamma\Phi}^H+\lambda\mathbf{I}\right)^{-1}\hat{\mathbf{R}}_{\mathbf{y}}\right),\label{eq:sbloptphiopt}
\end{equation}where we also allow all possible dictionaries $\mathbf{\Phi}$ with array manifold vectors as columns, to optimize the cost function. The following result follows similarly.\begin{theorem}
The problem in (\ref{eq:mlcost3}) and in (\ref{eq:sbloptphiopt}) are equivalent, in that they achieve the same globally minimum cost.\label{thm:sbleqprop}\vspace{-0.2cm}
\end{theorem}
\begin{proof}The proof follows similarly as for Theorem~\ref{thm:mleeqprop}, and we present it it Appendix section~\ref{prf:sbleqprop} for completion.
\end{proof}\noindent The above results help to understand the proposed approach in (\ref{eq:mlcost3}): (\ref{eq:mlcost3}) estimates a structured covariance matrix fit to the measurements in the MLE sense over all model orders for classical MLE or all appropriate dictionaries for SBL.\\\noindent The entries of a structured matrix and noise variance may be combined as presented in \cite{qiao17}. However, the choice of explicitly involving $\lambda$ parameter has two important consequences: a) If $\sigma_n^2$ is known, the proposed approach allows a mechanism to feed this information, which is absent in\cite{qiao17} b) If $\sigma_n^2$ is unknown, a better learning strategy to estimate the noise variance and then feeding it as part of the model may result in better DoA estimates than jointly estimating $\bm{\theta}$ and $\sigma_n^2$. Finally, although the optimization problem in \cite{qiao17} and the proposed are similar, an algorithm for solving it is missing in \cite{qiao17}. During the preparation of this manuscript we came across another recent work in \cite{yang22}
which derives from the \emph{classical} MLE formulation. Consequently, it involves the non-linear rank constraint which is implemented by a truncated eigen-decomposition step. In contrast, the presented algorithm builds on the success of SBL algorithm and relaxes the rank constraint, similar to SBL. The presented approach also guarantees that the likelihood increases over the iterations. $K$ is utilized for root-MUSIC. 
\cite{yang22} 
focuses on sparse linear arrays i.e., sensors on grid. We, however, consider the general non-uniform linear array case as well, and is discussed next.
\section{Gridless SBL with Likelihood-based Grid Refinement}\label{sec:sblgdrefine}
\noindent 
In this section, we consider the case when sensors may be placed arbitrarily on a linear aperture. The presented ideas can be extended to other shapes or higher dimensional (2D, 3D, etc.) geometries. Note in both the sections we assume that the sensor positions are known. Issues concerning calibration errors is not the focus here, and we request interested readers to check the relevant literature for tackling such issues\cite{trees02}.

\noindent Consider an array with sensor positions $\mathbb{P}=\{0,1,2.1,3.5$ $,4.7,10\}$. The difference coarray for this geometry is given by $\mathbb{D}=\{0,1,1.1,1.2,1.4,2.1,2.5,2.6,3.5,3.7,4.7,5.3$ $,6.5,7.9,9,10\}$. The structured received signal covariance matrix is neither Toeplitz, nor is sampled from a higher order Toeplitz matrix. This implies that (\ref{eq:mlcost3}), where we enforced a Toeplitz PSD constraint, is not applicable. Similarly, the second step wherein we estimate DoAs in a gridless manner using root-MUSIC is not applicable. Finally, the number of distinct lags is $\vert\mathbb{D}\vert=M(M+1)/2-(M-1)=16$ which essentially enforces structure only on the diagonal entries, in that, they be equal. This indicates poor availability of structural constraints on the received signal covariance, compared to the geometries where sensors are present on a uniform grid. Note that the SBL formulation in (\ref{eq:sblopt}) is devoid of such limitations. Using the recovered $\bm{\Gamma}^*$ one can construct a Toeplitz matrix of order $\lfloor M_{\mathrm{apt}}\rfloor$, where $\lfloor\cdot\rfloor$ indicates the floor function, and beyond although the accuracy may not be reliable as the measurements lack information about larger lags. This highlights the versatility with which SBL can handle arbitrary array geometries. 
However, as we already know that SBL does not quite solve (\ref{eq:parmodl}) that we ventured out to solve in the first place, because the DoAs may not lie on the chosen grid. One can employ a very fine grid, but the per iteration computational complexity increases linearly with the grid size. We extend the SBL procedure to progressively refine the initial uniform coarse grid by adding more points near potential source locations. We achieve this 
in two steps: (a) Grid point adjustment around peaks, in the solution $\bm{\gamma}^*$ of (\ref{eq:sblopt}) using sequential SBL\cite{tipping03} to simultaneously update both grid point and power estimate (b) Multi-resolution grid refinement. Note that the latter builds on the former step by re-running SBL after the local 
step (a), pruning, and increasing grid resolution near top peaks in the $\bm{\gamma}$ pseudospectrum. As will be shown next, the grid point adjustment around peaks in $\bm{\gamma}$ pseudospectrum is a computationally simpler procedure to further increase the likelihood after SBL iterations on a coarse grid.\subsection{Grid Point Adjustment around Peaks in Solution $\bm{\gamma}^*$ of (\ref{eq:sblopt})}
\noindent We begin by rewriting the SBL objective function to separate out the $i$-th grid component characterized by the tuple $(\gamma_i,u_i)$; $u=\sin\theta$ is used here. Let $\mathbf{C}=\mathbf{\Phi}\bm{\Gamma}\mathbf{\Phi}^H+\lambda\mathbf{I}$ and $\mathbf{C}_{-i}=\mathbf{\Phi}_{-i}\bm{\Gamma}_{-i}\mathbf{\Phi}_{-i}^H+\lambda\mathbf{I}$, where $\mathbf{\Phi}_{-i}$ denotes the dictionary without the $i$-th column in $\mathbf{\Phi}$, and $\bm{\Gamma}_{-i}$ denotes the matrix without the
the $i$-th row and the $i$-th column 
in $\bm{\Gamma}$. Then\begin{equation}
\mathcal{L}(\bm{\gamma})=\log\det\mathbf{C}+\mathrm{tr}\left(\mathbf{C}^{-1}\hat{\mathbf{R}}_{\mathbf{y}}\right)
=\mathcal{L}(\bm{\gamma}_{-i})+L(\gamma_i,u_i),\label{eq:sbloptisplit}
\end{equation}where $\mathcal{L}(\bm{\gamma}_{-i})=\log\det\mathbf{C}_{-i}+\mathrm{tr}\left(\mathbf{C}_{-i}^{-1}\hat{\mathbf{R}}_{\mathbf{y}}\right)$ is devoid of $(\gamma_i,u_i)$, and $L(\gamma_i,u_i)=\log(1+\gamma_i\mathbf{\Phi}_i^H\mathbf{C}_{-i}^{-1}\mathbf{\Phi}_i)-\frac{\mathbf{\Phi}_i^H\mathbf{C}_{-i}^{-1}\hat{\mathbf{R}}_{\mathbf{y}}\mathbf{C}_{-i}^{-1}\mathbf{\Phi}_i}{\gamma_i^{-1}+\mathbf{\Phi}_i^H\mathbf{C}_{-i}^{-1}\mathbf{\Phi}_i}$ (see eq.~(18) in \cite{tipping03} for detailed derivation of (\ref{eq:sbloptisplit})). Let $i_k^{(0)},k\in\{1,\ldots,K\}$ denote the indices for the $K$ top peaks in the $\bm{\gamma}$ pseudospectrum. The index superscript $(.)$ in $i_k^{(0)}$ indicates the iteration number of the overall two step procedure, and will be discussed more in the next subsection. The idea is to fix the first term in RHS of (\ref{eq:sbloptisplit}) and minimize the objective $L(\gamma,u)$ with respect to $(\gamma,u),u\in[u_{i_k^{(0)}}-\delta,u_{i_k^{(0)}}+\delta],k=\{1,\ldots,K\}$, one peak at a time; the bound\footnote{Note that future iterations may involve non-uniform grid, and a similar bound on either directions is used to avoid grid point overlap.} $\delta<1/G$ is to avoid grid point overlap. In other words, the aim is to solve\begin{equation}
\underset{u\in[u_{i_k^{(0)}}-\delta,u_{i_k^{(0)}}+\delta]}{\min}\>\>\>\underset{\gamma\geq 0}{\min}\>L(\gamma,u)=\log(1+\gamma s(u))-\frac{q(u)}{\gamma^{-1}+s(u)}\label{eq:neighopt}\IEEEeqnarraynumspace\end{equation}where $q(u)=\bm{\phi}(u)^H\mathbf{C}_{-i}^{-1}\hat{\mathbf{R}}_{\mathbf{y}}\mathbf{C}_{-i}^{-1}\bm{\phi}(u)$ and $s(u)=\bm{\phi}(u)^H\mathbf{C}_{-i}^{-1}\bm{\phi}(u)$. The minimization with respect to $\gamma$ for a fixed $u$ can be obtained in closed-form as\begin{equation}
\gamma_{\mathrm{opt}}(u)=\left\{\begin{array}{cc}
            \frac{q(u)-s(u)}{s(u)^2} & q(u)>s(u) \\
            0 & q(u)\leq s(u)
        \end{array}\right..
\end{equation}And thus we have\begin{equation}
        L(\gamma_{\mathrm{opt}}(u),u)=\left\{\begin{array}{cc}
            \log\left(\frac{q(u)}{s(u)}\right)-\frac{q(u)}{s(u)}+1 & q(u)>s(u) \\
            0 & q(u)\leq s(u)
        \end{array}\right..
\end{equation}Note that $\log\left(\frac{q(u)}{s(u)}\right)-\frac{q(u)}{s(u)}+1\leq 0,\forall u$, and is equal to zero only when $q(u)=s(u)$ for some $u$.
\begin{figure}
    \centering
    \includegraphics[width=\linewidth]{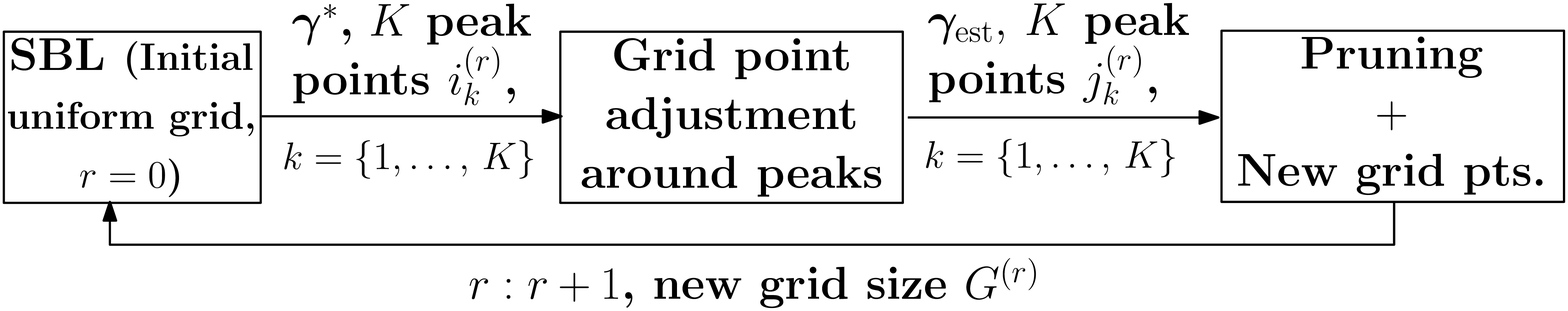}
    \caption{Proposed SBL with likelihood-based grid refinement procedure}
    \label{fig:twostepproc}
\end{figure}Consequently, we are interested in $u\in[u_{i_k^{(0)}}-\delta,u_{i_k^{(0)}}+\delta]$ such that $q(u)>s(u)$. For such points $L(\gamma_{\mathrm{opt}}(u),u)$ is a monotonic non-increasing function of `$\frac{q(u)}{s(u)}$', and thus the problem in (\ref{eq:neighopt}) reduces to the following problem\begin{equation}
u^*=\underset{u\in[u_{i_k^{(0)}}-\delta,u_{i_k^{(0)}}+\delta]\mbox{ s.t. }q(u)>s(u)}{\arg\max}\quad R(u)=\frac{q(u)}{s(u)}.\label{eq:neighopt2}
\end{equation}We provide the following perspective to understand the objective we wish to locally maximize.\begin{IEEEeqnarray}{ll}
R(u)&=\frac{q(u)}{s(u)}=\frac{\bm{\phi}(u)^H\mathbf{C}_{-i}^{-1}\hat{\mathbf{R}}_{\mathbf{y}}\mathbf{C}_{-i}^{-1}\bm{\phi}(u)}{\bm{\phi}(u)^H\mathbf{C}_{-i}^{-1}\bm{\phi}(u)}\nonumber\\
&=\frac{\bm{\phi}(u)^H\mathbf{C}_{-i}^{-1}\hat{\mathbf{R}}_{\mathbf{y}}\mathbf{C}_{-i}^{-1}\bm{\phi}(u)/\left(\bm{\phi}(u)^H\mathbf{C}_{-i}^{-1}\bm{\phi}(u)\right)^2}{1/\left(\bm{\phi}(u)^H\mathbf{C}_{-i}^{-1}\bm{\phi}(u)\right)},\IEEEeqnarraynumspace
\end{IEEEeqnarray}which is the ratio of (numerator) actual beamforming total output power and (denominator) expected beamforming interference plus noise output power, where the model interference plus noise signal covariance is given by $\mathbf{C}_{-i}$. The expression utilizes the beamformer $\mathbf{w}=\frac{\mathbf{C}_{-i}^{-1}\bm{\phi}(u)}{\left(\bm{\phi}(u)^H\mathbf{C}_{-i}^{-1}\bm{\phi}(u)\right)}$, which represents a minimum variance distortionless response (MVDR) beamformer with $\mathbf{C}_{-i}$ as the model interference plus noise signal covariance matrix\cite{trees02,alshoukairi21}. 
Thus the criterion $R(u)$ in (\ref{eq:neighopt2}) picks a $u$ in the neighbourhood of $u_{i_k^{(0)}}$ that most exceeds the expected beamforming interference plus noise output power, guided by $\mathbf{C}_{-i}$. At the true location, which is likely to be in the search region, the model $\mathbf{C}_{-i}$ expects low power but hopefully the measurements indicate higher power than expected.

\noindent We solve (\ref{eq:neighopt2}) by implementing a fine grid of size $G'$
around the peak and evaluating the criterion $R(u)$. Once we find the maximum point we replace the grid point $(\gamma_{i_k^{(0)}},u_{i_k^{(0)}})$ with $(\gamma_{\mathrm{opt}}(u^*),u^*)$. Note that by including the previous grid point in the search region we ensure that the likelihood is steadily increasing. We then repeat this procedure for the next peak, corresponding to another source, and so on. This procedure around a peak \emph{assumes} that other peaks were reliably estimated, which may hold only approximately. Therefore, we iterate over the $K$ peaks until convergence. In practice, the procedure converges quickly over $20-30$ iterations.

\noindent The above procedure is quite different from that in \cite{malioutov05}, in that it offers a means to improve the DoA estimate without requiring to re-run the primary procedure (here: SBL, in \cite{malioutov05}: $\ell_1$-SVD) on all grid points. As will be shown numerically in Section~\ref{sec:gridrefinestudy}, this step alone improves the solution significantly.\begin{figure*}[h]
\centering
\begin{tabular}{c@{\hskip -0.1cm}c@{\hskip -0.1cm}c}
\includegraphics[width=0.33\linewidth]{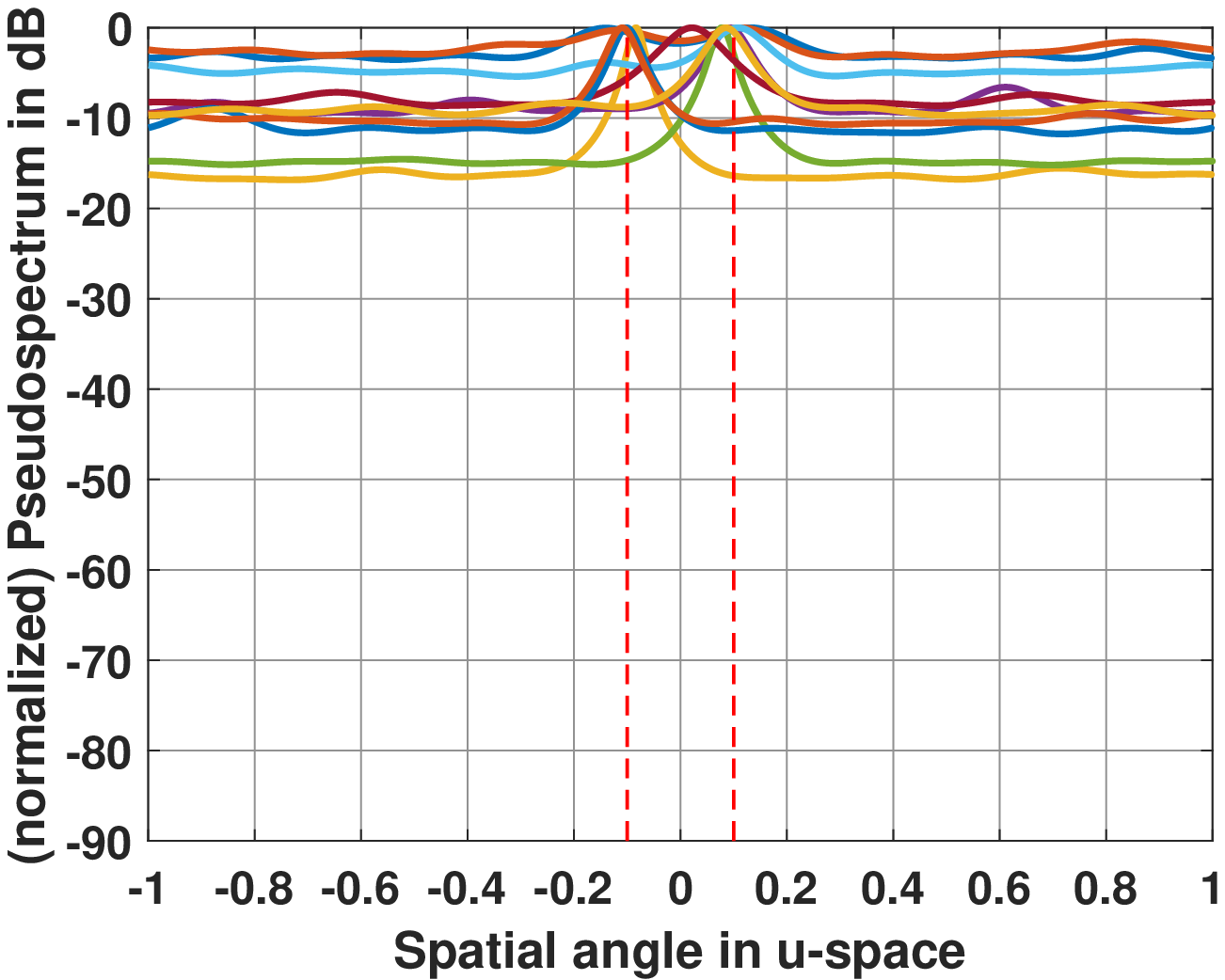} & \includegraphics[width=0.33\linewidth]{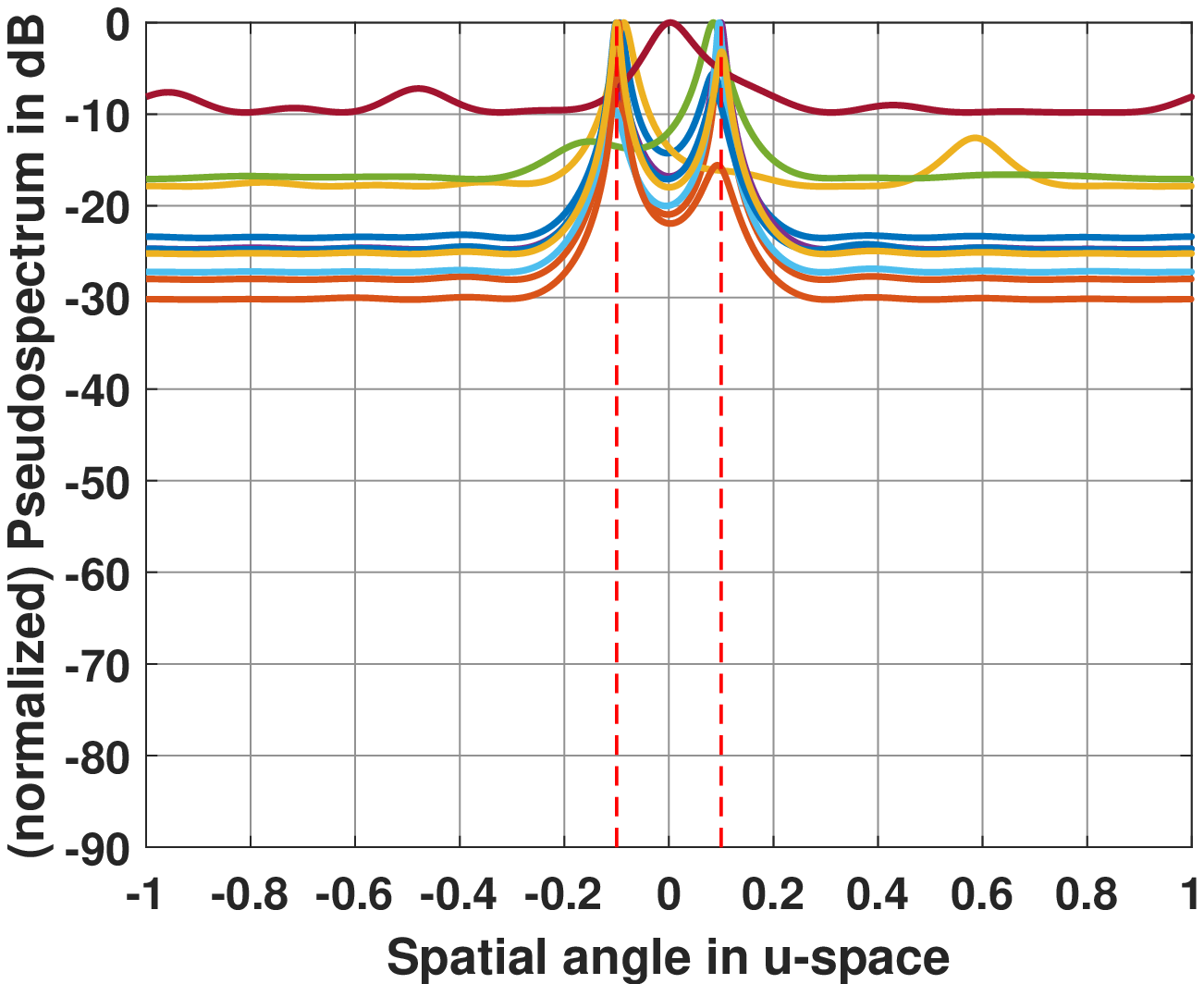} & \includegraphics[width=0.33\linewidth]{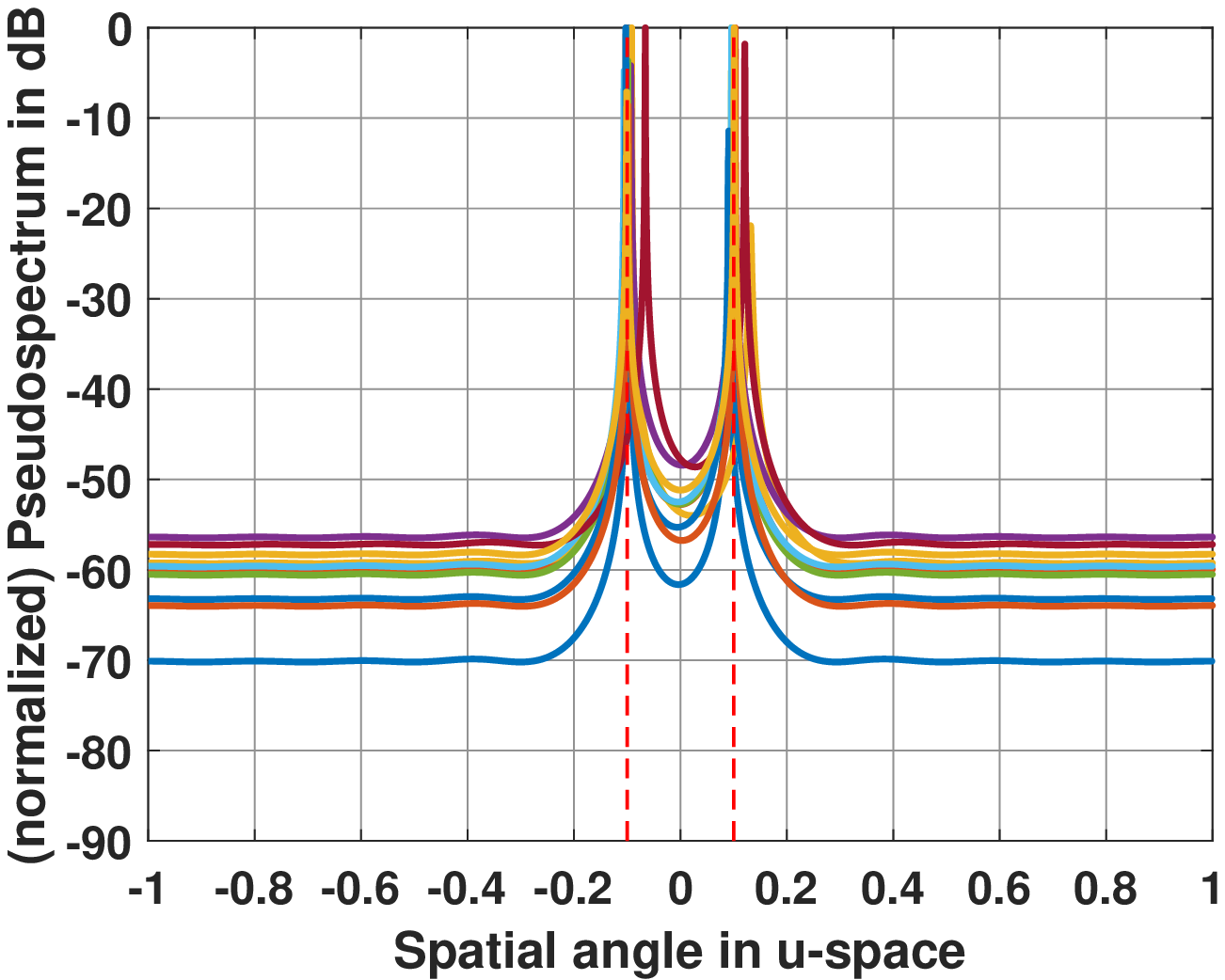}\\(a) SCM & (b) FB averaged SCM & (c) StructCovMLE (Proposed)\end{tabular}
\caption{Single snapshot scenario}\vspace{-0.5cm}\label{fig:singlesnap}
\end{figure*}\begin{figure*}[h]
\centering
\begin{tabular}{c@{\hskip -0.2cm}c@{\hskip -0.2cm}c@{\hskip -0.2cm}c}
\includegraphics[width=0.26\linewidth]{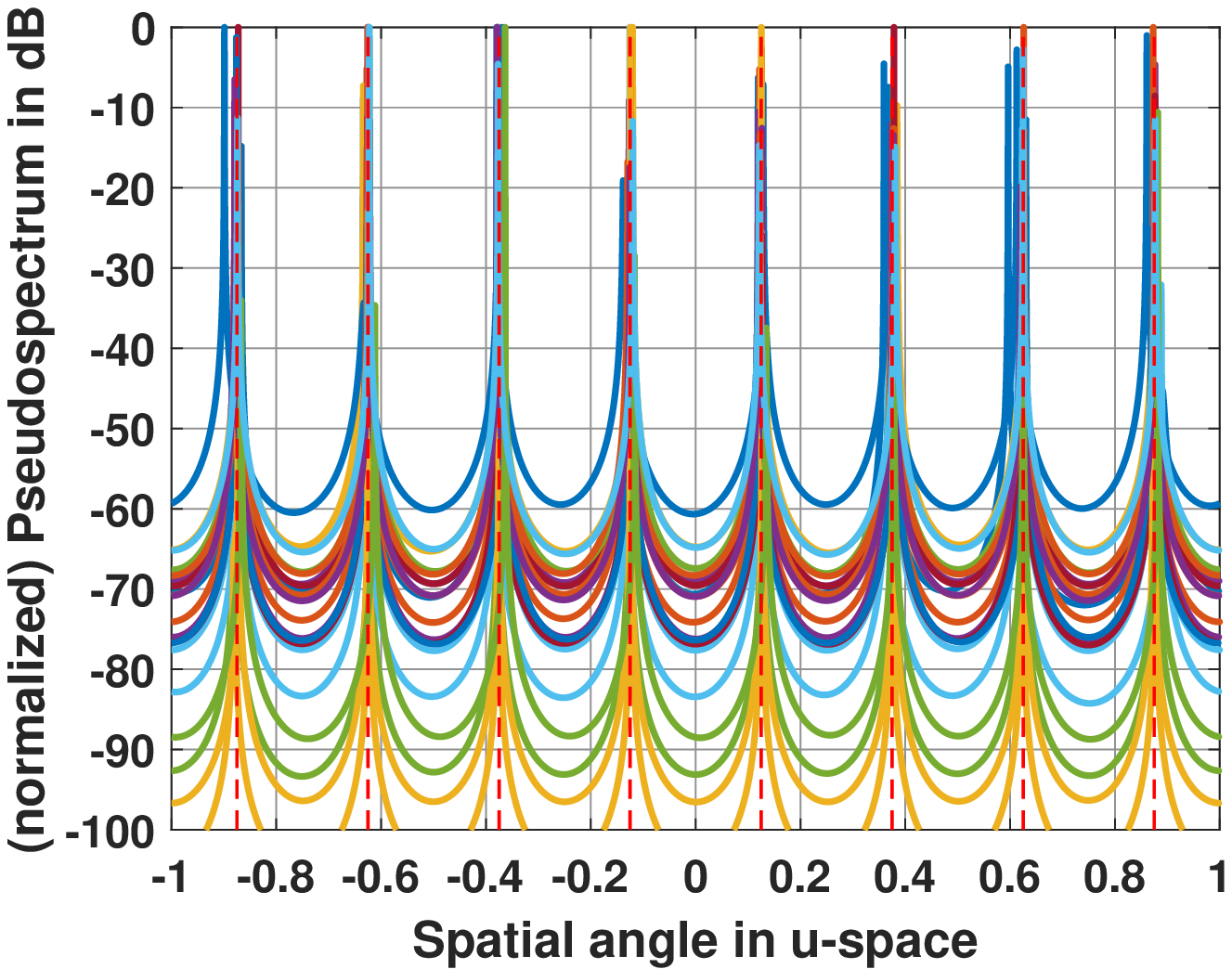}
&\includegraphics[width=0.26\linewidth]{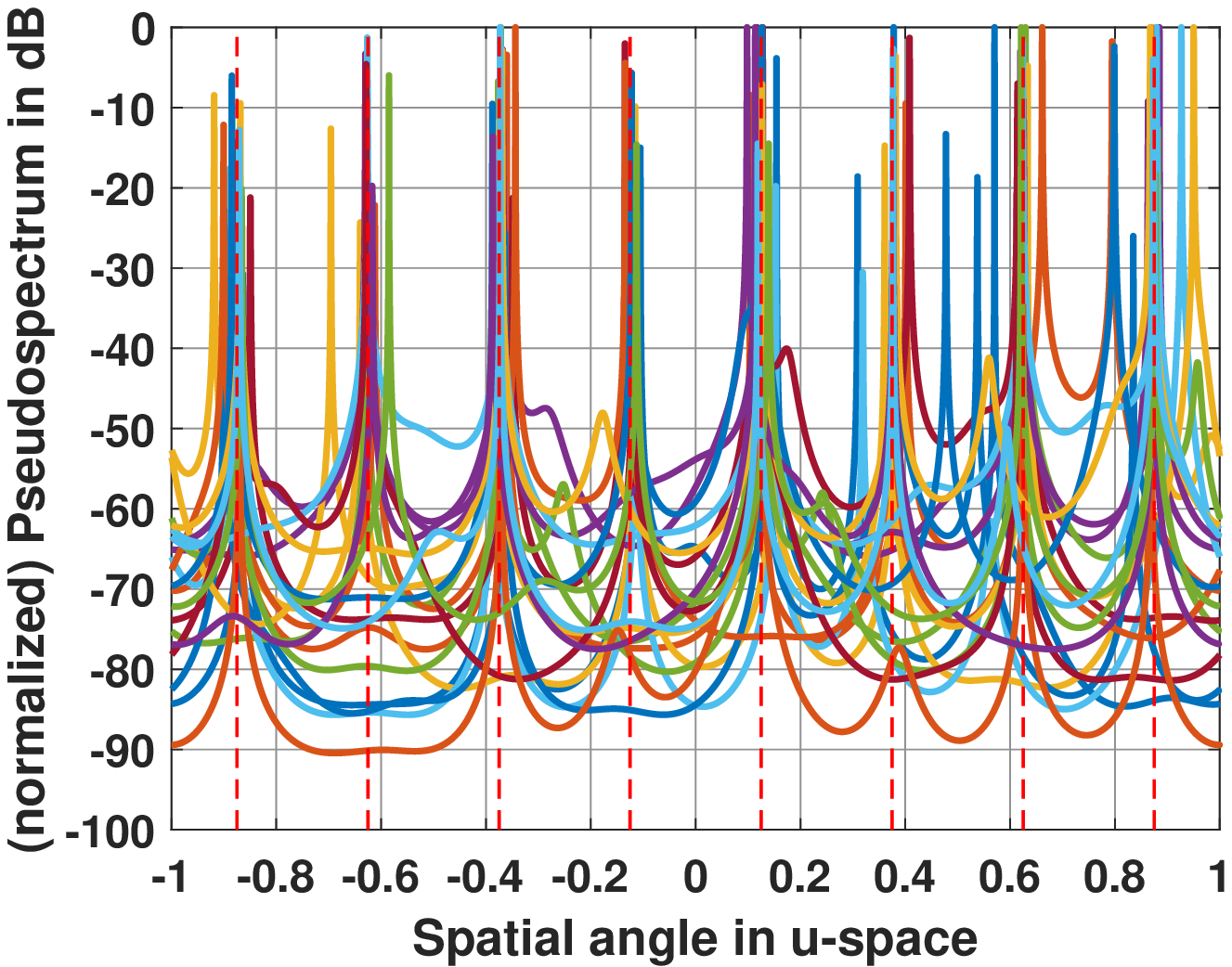} &\includegraphics[width=0.26\linewidth]{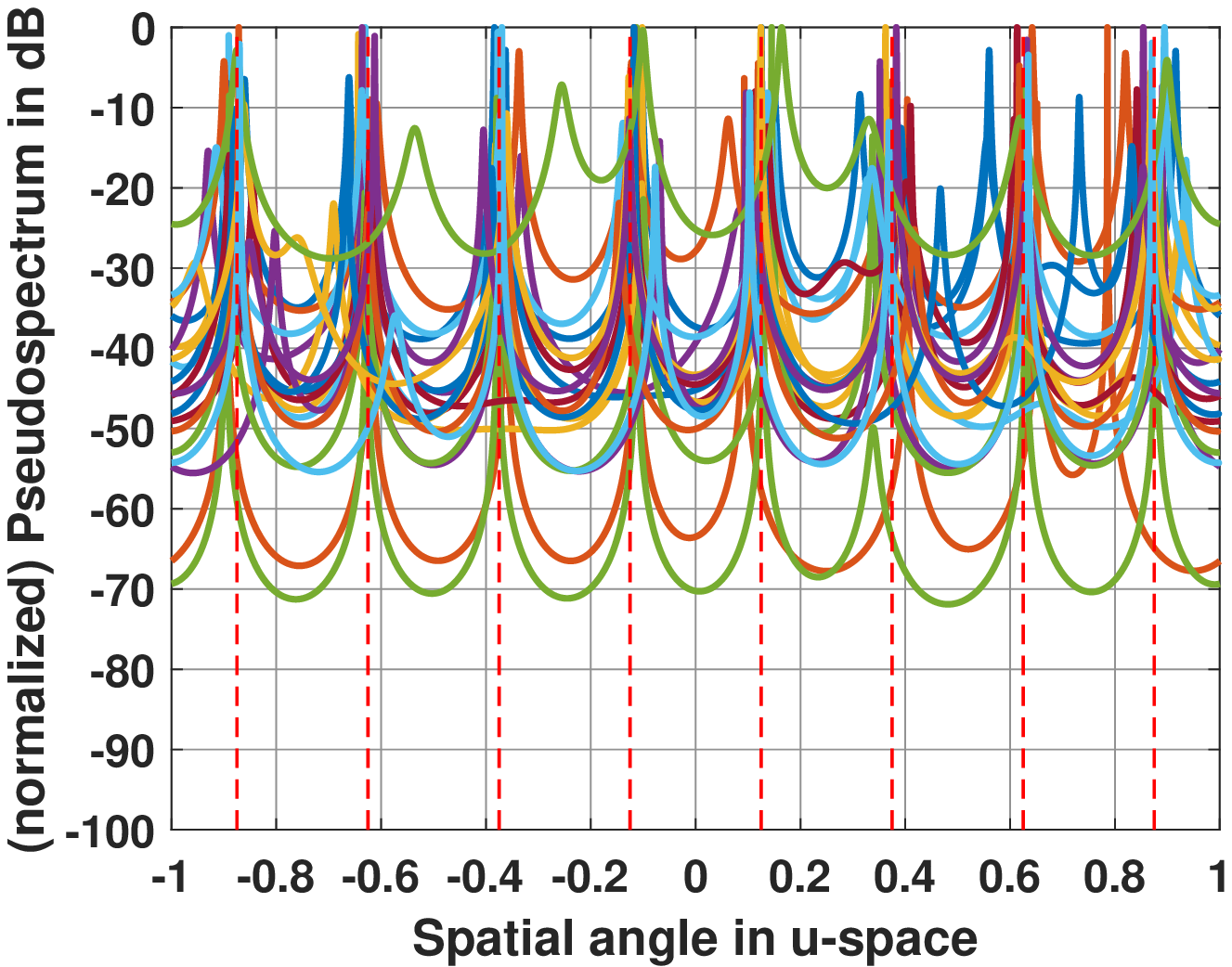} & \includegraphics[width=0.26\linewidth]{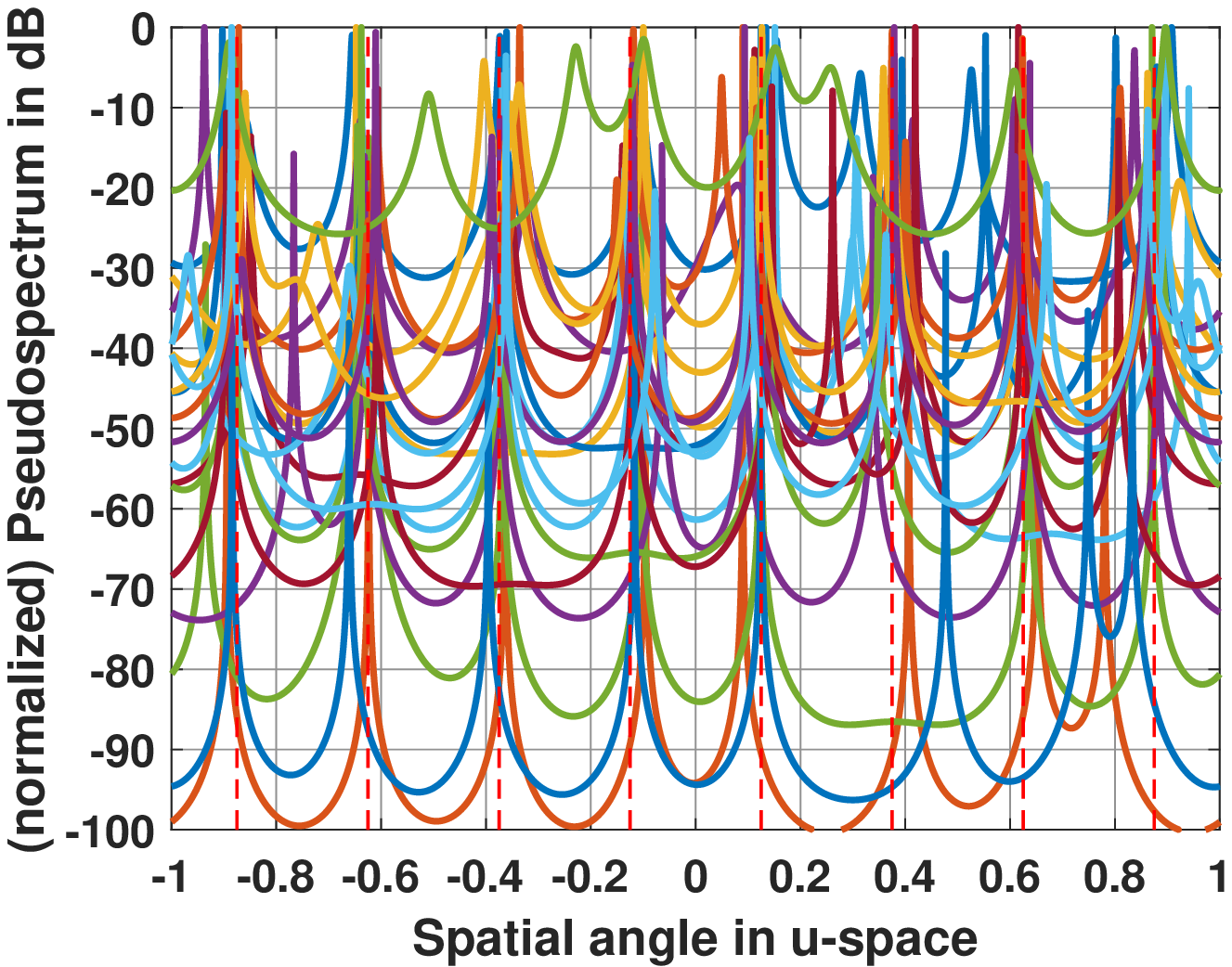}\\
\begin{tabular}{c}(a) StructCovMLE\\(RMSE = $0.005$)\end{tabular} & \begin{tabular}{c}(b) RAM\\(RMSE = $0.106$)\end{tabular} & \begin{tabular}{c}(c) GL-SPARROW\\(RMSE = $0.050$)\end{tabular} & \begin{tabular}{c}(d) GL-SPICE\\(RMSE = $0.072$)\end{tabular}\end{tabular}
\caption{More sources than sensors' case}\vspace{-0.5cm}\label{fig:moresources}
\end{figure*}\subsection{Multi-resolution Grid Refinement}
\noindent In this subsection, we take the local grid point adjustment step further by introducing a finer grid around the new peak locations and by re-running the SBL procedure. Let $r=0$ and $j_k^{(r)},k\in\{1,\ldots,K\}$, denote the indices for the $K$ top peaks in the solution $\bm{\gamma}_{\mathrm{est}}$ after grid point adjustment around the peaks, and $G^{(0)}=G$ denotes the current grid size. We run the following procedure to further refine the solution:\begin{enumerate}[leftmargin=*]
    \item Prune the grid points $i$ for $\{i:\bm{\gamma}_{\mathrm{est}}(i)<\gamma_{\mathrm{thresh}}\}$ for some $\gamma_{\mathrm{thresh}}$. In the simulations we set $\gamma_{\mathrm{thresh}}=10^{-3}$.
    \item Introduce new grid points in the region $[u_{j_k^{(0)}}-\frac{4}{G^{(0)}},u_{j_k^{(0)}}+\frac{4}{G^{(0)}}]$ with finer resolution $\frac{1}{g}\frac{2}{G^{(0)}},g>1$. The region includes two neighbouring grid points on each side. `$g$' is chosen such that the total number of grid points does not exceed $G^{(0)}$. This choice ensures that per iteration complexity is contained. In the simulations, we choose $g=3$; in general the procedure adds $(4g+1)$ new points per peak of interest.
    \item Increment $r:r+1$, and update the grid size $G^{(r)}$. Run SBL from scratch, get new set of indices for $K$ top peaks $i_k^{(r)},k\in\{1,\dots,K\}$; perform grid point adjustment at these peaks to get $j_k^{(r)},k\in\{1,\ldots,K\}$, as updated peak locations. Go to step 1).
\end{enumerate}This procedure is similar to that in \cite{malioutov05}. We run these steps a few times and report the peak points as DoA estimates in the simulation section. A natural question that arises is: \emph{why is the local grid point adjustment not enough for improved resolution and separating two closely spaced sources?} In other words, \emph{is a SBL re-run necessary?}\\\noindent A re-run improves both, the SBL with coarse grid and the grid point adjustment around peaks. For the SBL procedure, a re-run in this manner provides a much more informed sampling of the spatial coordinates with closely-spaced grid points around locations of interest. For the grid point adjustment step where the MVDR beamformer is employed, a finer grid helps to further ensure that only a single source is present in the search region of (\ref{eq:neighopt2}). This is important because the beamformer $\mathbf{w}=\frac{\mathbf{C}_{-i}^{-1}\bm{\phi}(u)}{\left(\bm{\phi}(u)^H\mathbf{C}_{-i}^{-1}\bm{\phi}(u)\right)}$ engages its degrees of freedom and attempts to null interference outside of this search region, and the criterion $R(u)$ works best only if a single source is present in the search region. A block diagram summarizing the high level steps suggested in this section for the general case of sensors being placed arbitrarily is shown in Fig.~\ref{fig:twostepproc}.\vspace{-0.2cm}
\section{Simulation Results}\label{sec:sim}
\noindent We present numerical results to evaluate the performance of the proposed algorithms in Section~\ref{sec:structcovmle} and \ref{sec:sblgdrefine} in different scenarios. We also compare the proposed `StructCovMLE' algorithm with MUSIC using SCM, MUSIC using forward-backward (FB) averaged SCM\cite{trees02}, reweighted ANM (RAM), GridLess (GL)-SPICE, GL-SPARROW and 
Cram\'er-Rao bound (CRB)\cite{stoica01}. We initialize all the iterative techniques (i.e. the proposed Algorithm~\ref{alg:glssr} and RAM) with the unit vector $\mathbf{v}_0=\mathbf{e}_1$ for reasons stated in remark~\ref{rem:initial}, and run $20$ iterations 
unless otherwise specified. We provide the number of sources, $K$, to identify to all the algorithms. We
set $\lambda=\sigma_n^2$ for the proposed algorithm. 
The RAM implementation follows its description in \cite{yang16}, and we also adopt the dimension reduction mechanism suggested in the paper. We set $\eta=\sigma_n\sqrt{ML+2\sqrt{ML}}$ as suggested for DoA estimation in \cite{yang16}. For GL-SPARROW, we set $\lambda=\sigma_n\sqrt{M\log M}$ as suggested in \cite{steffens18,wassim17}. GL-SPICE does not require or utilize the knowledge of the noise variance, $\sigma_n^2$. We compute root mean squared error (RMSE) in $u$-space ($u=\sin\theta$) as\begin{equation}
    \mbox{RMSE}=\sqrt{\frac{1}{T}\frac{1}{K}\sum_{t=1}^T\sum_{k=1}^K\left(\hat{u}_{k,t}-u_{k}\right)^2},
\end{equation}where $T$ denotes the total number of random trials.\vspace{-0.2cm}
\begin{figure*}[h]
\centering
\begin{tabular}{ccc}
\includegraphics[width=0.3\linewidth]{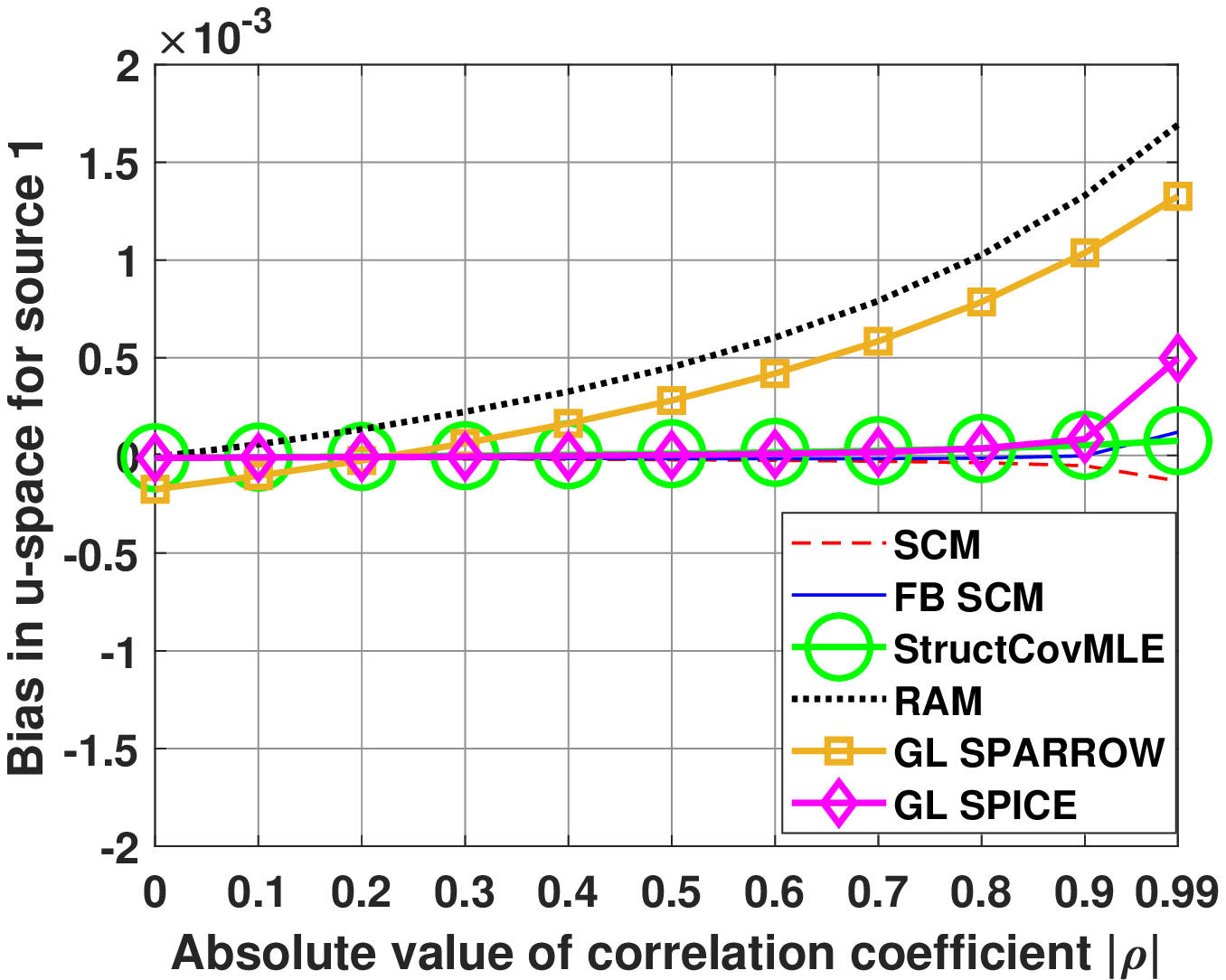} & \includegraphics[width=0.3\linewidth]{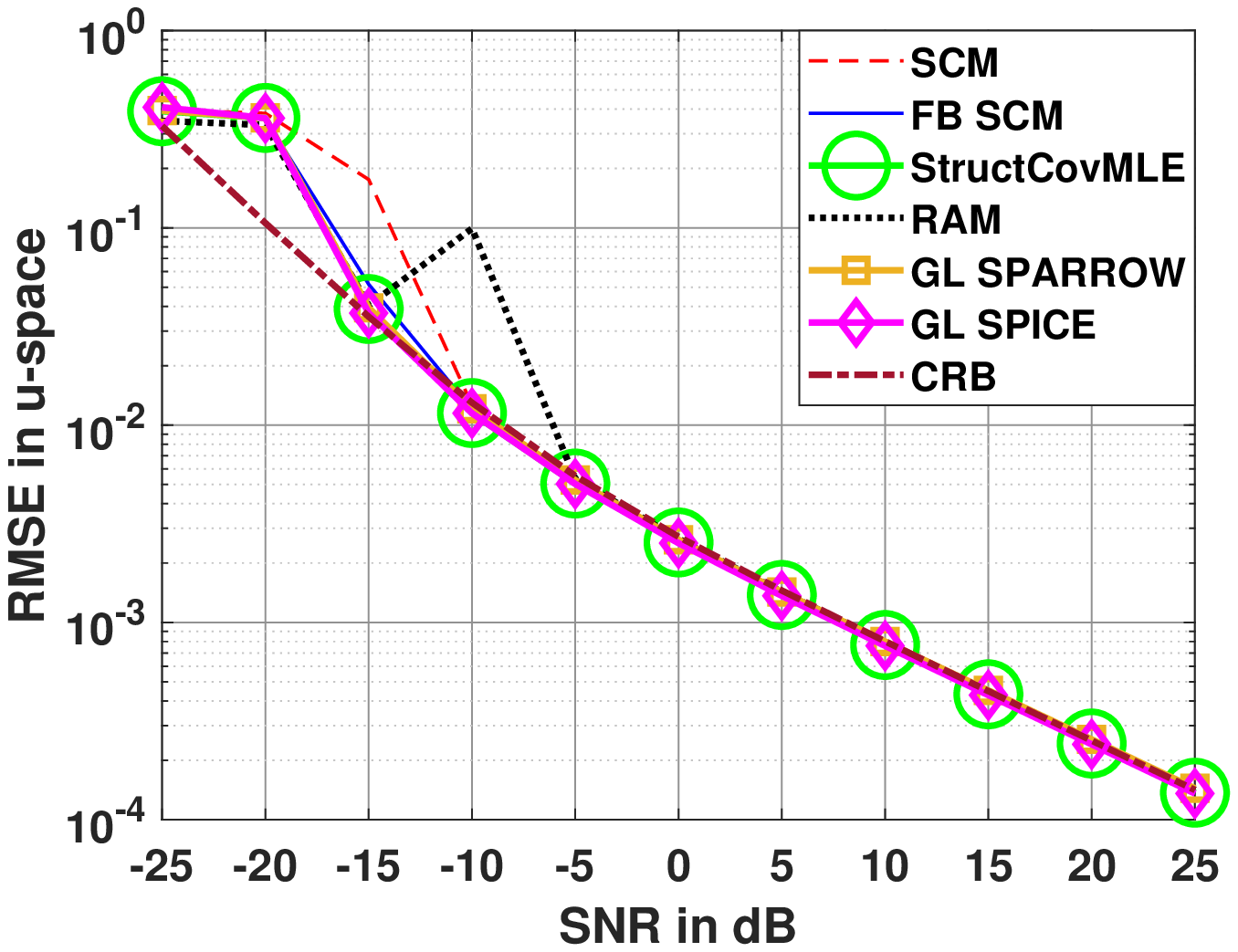} & \includegraphics[width=0.3\linewidth]{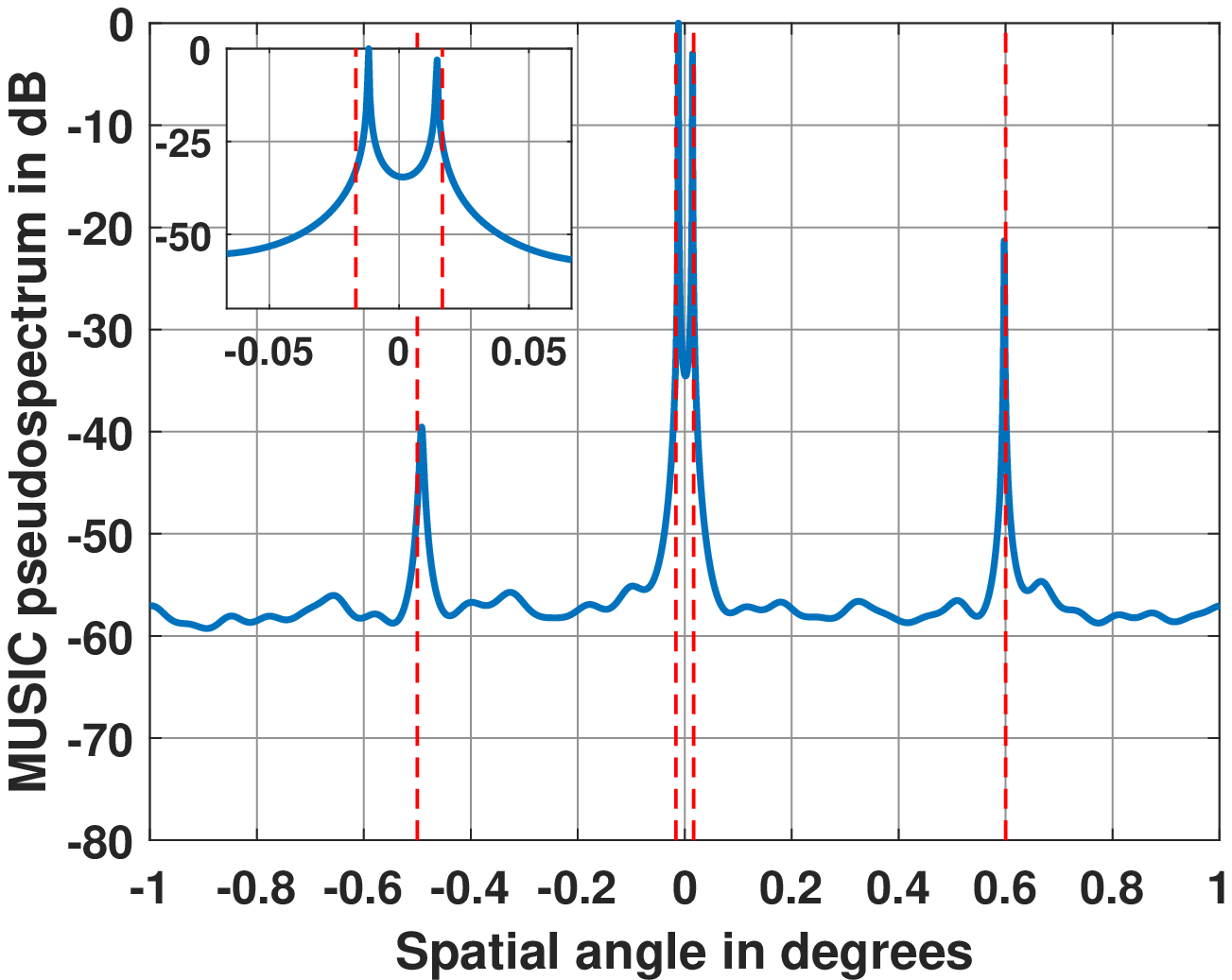}\\(a) Source 1 & (c) Uncorrelated sources & (e) StructCovMLE\\\includegraphics[width=0.3\linewidth]{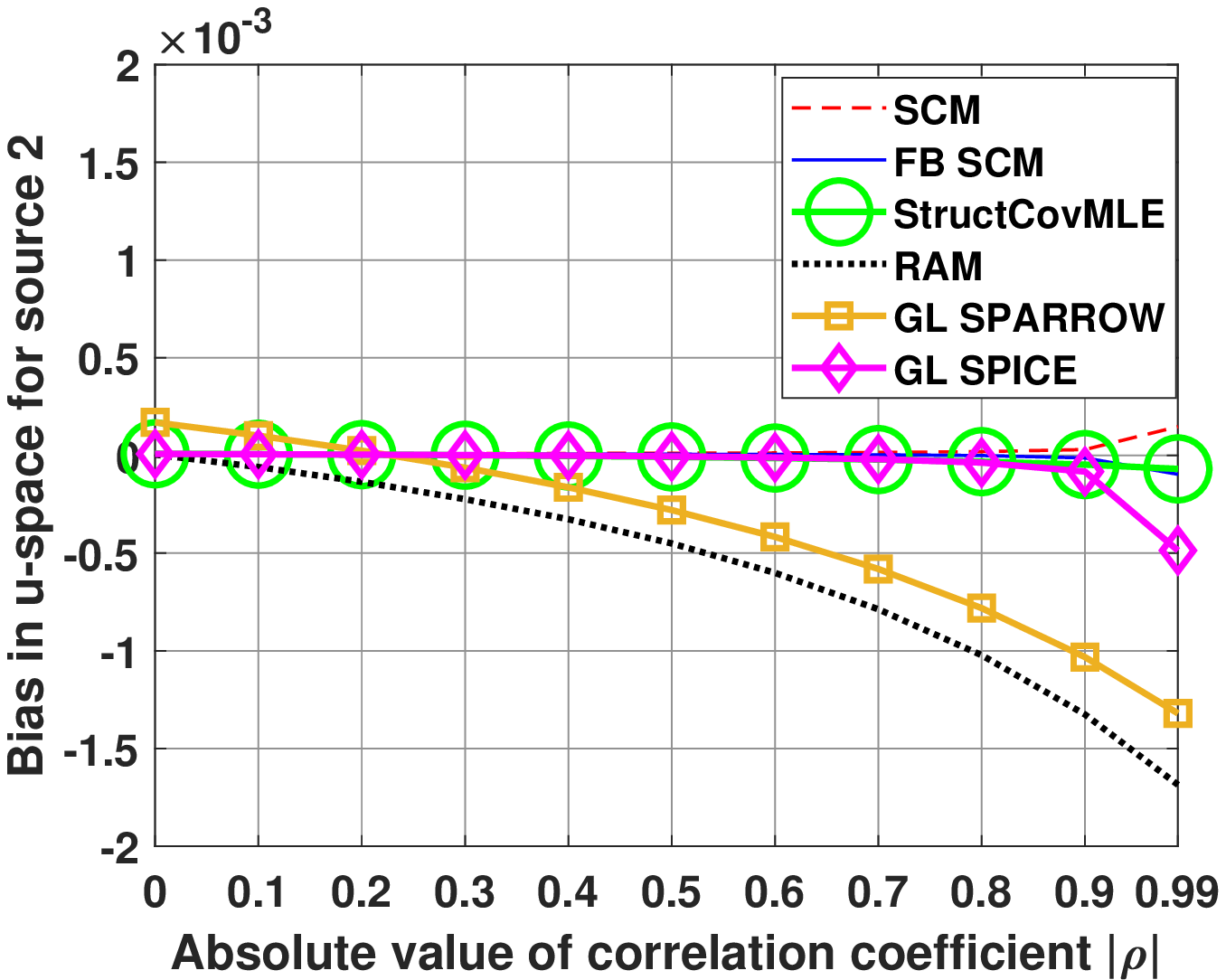} & \includegraphics[width=0.3\linewidth]{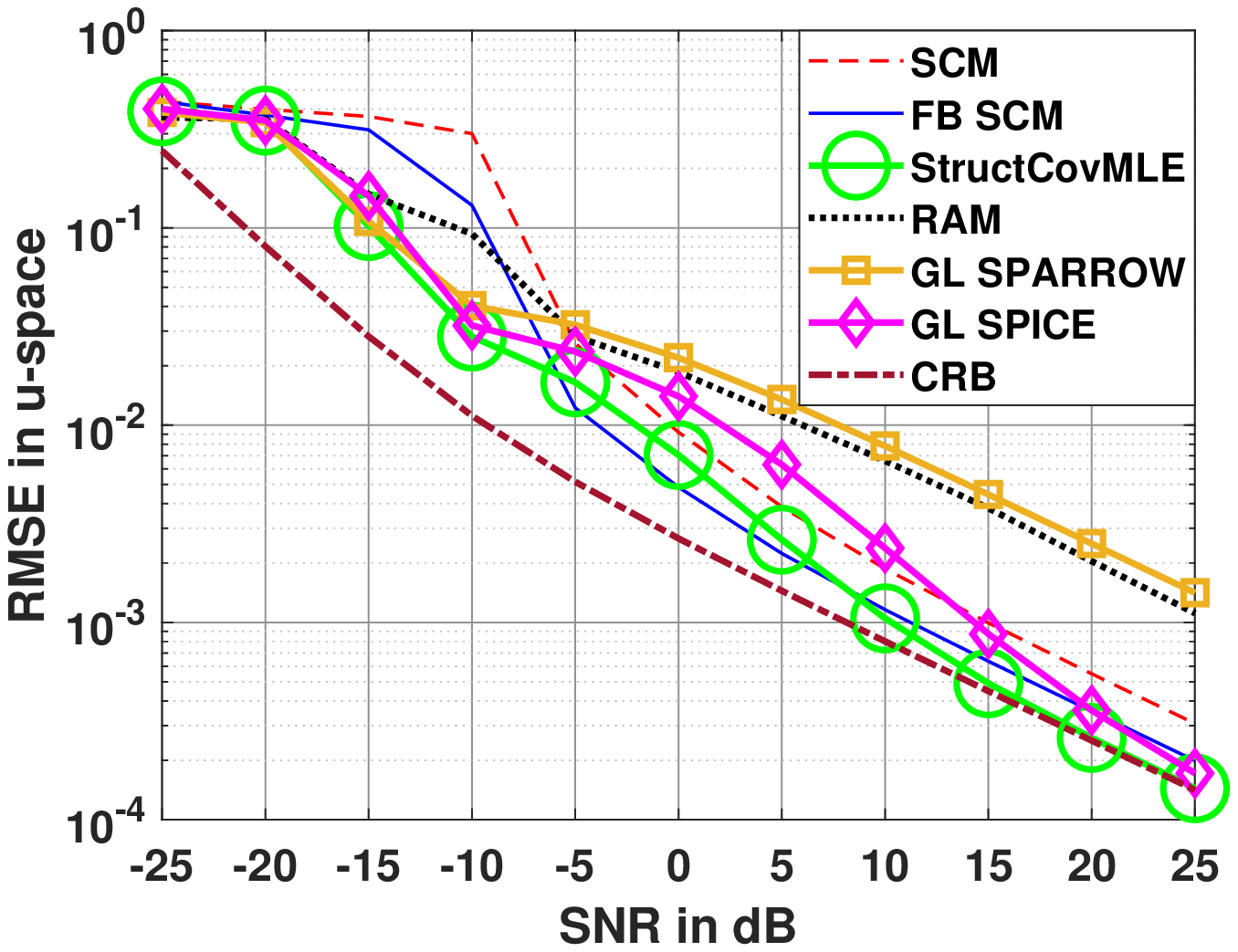} & \includegraphics[width=0.3\linewidth]{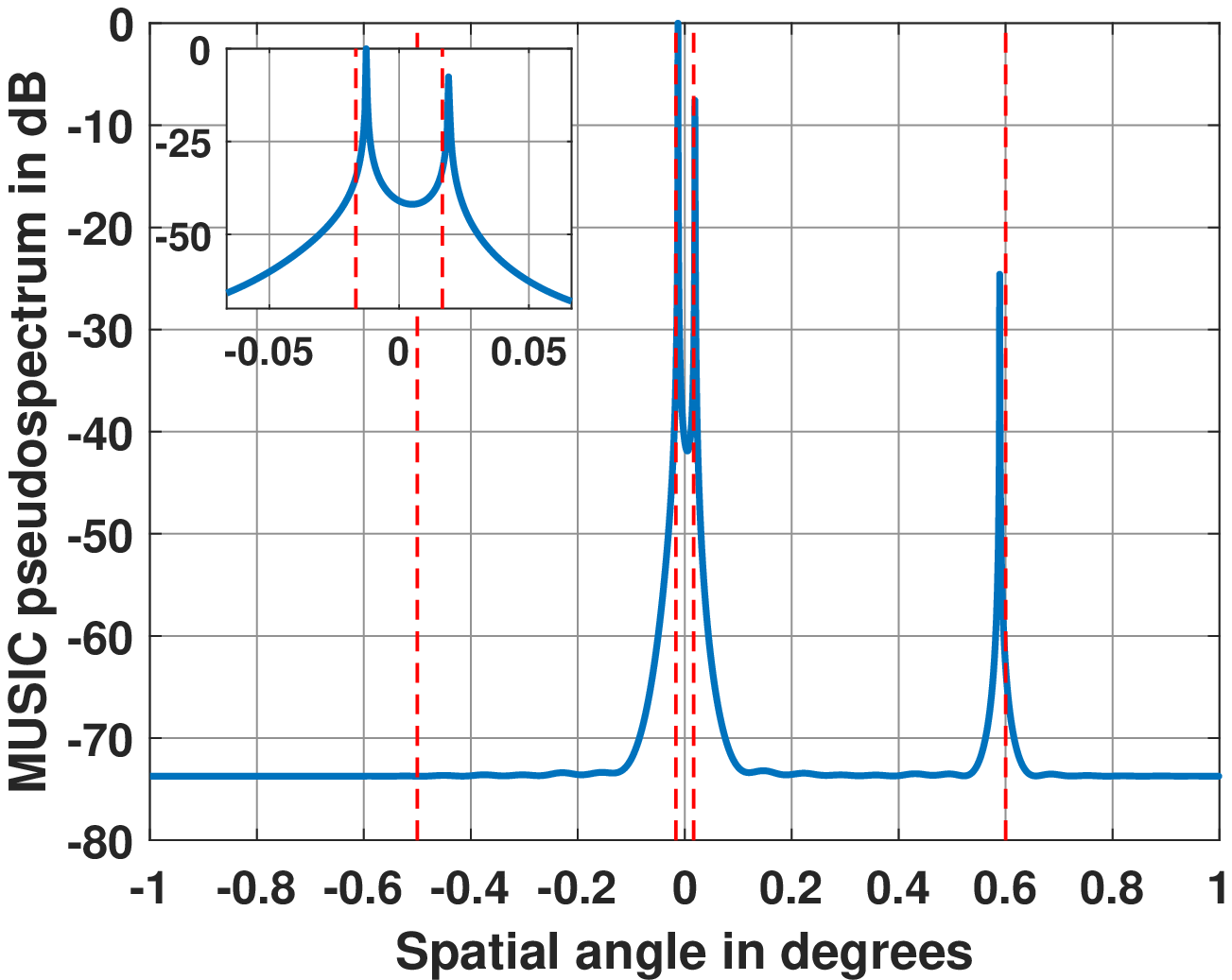}\\(b) Source 2 & (d) Correlated sources, $\vert\rho\vert=0.9$ & (f) RAM\end{tabular}
\caption{(a) \& (b): Effect of correlation ($\rho/\vert\rho\vert=0.5010+j0.8654$) on empirical bias. (c) \& (d): RMSE as a function of SNR. (e) \& (f): Nested array with sensor locations, $\mathbb{P}=\{0,1,2,3,4,5,11,17,23,29\}$. }\vspace{-0.5cm}\label{fig:combinedstudy}
\end{figure*}\subsection{Performance in Single Snapshot Case i.e., $L=1$}\label{sec:singlesnap}
\noindent We consider a ULA with $M=10$ sensors, and two sources at angles $\{-1/M_{\mathrm{apt}},1/M_{\mathrm{apt}}\},M_{\mathrm{apt}}=10$, in u-space. The sources have signal-to-noise ratio (SNR) $=20$ dB. In Fig.~\ref{fig:singlesnap} we plot the normalized (with respect to maximum value) MUSIC pseudospectrum for different estimates of the structured covariance matrix, for $10$ random realizations. The true DoAs are marked in vertical red dashed curve. As expected, the SCM provides the worst performance as it does not satisfy the rank requirement to identify two sources. It is observed that the forward-backward averaging helps to improve the performance by ensuring the rank requirement is satisfied. Still the overall performance is poor, as evident from the high noise floor compared to that for `StructCovMLE' in Fig.~\ref{fig:singlesnap} (c), and from its inability to resolve the two sources for some realizations. The proposed method is able to resolve both the sources. This is also true for RAM, GL-SPARROW, and GL-SPICE methods although we skip the plots in the interest of space.\vspace{-0.2cm}\subsection{More Sources than Sensors' Case}\label{sec:moresources}
\noindent We consider a nested array\cite{pal10} with $M=6$ sensors at locations $\{0,1,2,3,7,11\}$, and $K=8$ sources at angles uniformly in $u$-space. Their locations in MATLAB notation are $\{-1+1/K:2/K:1-1/K\}$. The SNR for each source is $20$ dB and
$L=4$. In Fig.~\ref{fig:moresources}, we plot $20$ random realizations of the normalized MUSIC pseudospectrum for the different estimates of the structured covariance matrix. As seen in Fig.~\ref{fig:moresources} (a), the proposed algorithm is able to localize all the $8$ sources, whereas the rest of the algorithms suffer from poor identifiability for some realizations. The superior performance is also evident from the lower RMSE value (in u-space) for the proposed algorithm, as compared to the other techniques.\vspace{-0.2cm}
\subsection{Effect of Correlation: An Empirical Bias Study}\label{sec:corrstudy}
\noindent We consider a ULA with $M=6$ sensors and two sources incoming at angles $\{-1/4,1/4\}$. The SNR is $20$ dB and 
$L=500$. In Fig.~\ref{fig:combinedstudy} (a) and (b) we plot the empirical bias for the two sources, respectively i.e., $\frac{1}{T}\sum_{t=1}^T\left(\hat{u}_{k,t}-u_k^*\right),k=\{1,2\}$, as a function of the absolute value of correlation coefficient, $\vert\rho\vert$. For computing the bias, we average over $T=50$ realizations. As observed in the plots, there is an increasing empirical bias in the angular estimates for the RAM and the GL-SPARROW techniques. This is evident as the curves drift away from the x-axis as $\vert\rho\vert$ increases. The proposed approach (shown in green curve with circular markers) has low empirical bias even when $\vert\rho\vert$ is as high as $0.99$. This demonstrates the superiority of the MLE based proposed approach over the other algorithms when there may be sources that are arbitrarily correlated.

\noindent Next, we provide RMSE vs. SNR curves for uncorrelated and correlated sources' case, and compare the performance with CRB. Note that in certain scenarios the algorithms may be biased, for example in extremely low SNR regime, or as evident in Fig.~\ref{fig:combinedstudy} (a) and (b) for RAM and GL SPARROW even in the high SNR scenario when the sources are correlated. We provide the curves for completion, but note that the CRB may not be a valid bound for such extreme cases.\vspace{-0.2cm}
\subsection{Performance as a Function of SNR}\label{sec:snrstudy}
\noindent We consider ULA with $M=6$ sensors and two sources at angles $\{-1/2,1/2\}$. 
$L=500$ and we run $30$ iterations for RAM and the proposed `StructCovMLE' algorithm. In Fig.~\ref{fig:combinedstudy} (c) and (d) we plot the RMSE (averaged over $T=50$ realizations) for the uncorrelated sources and correlated sources' cases, respectively, as a function of SNR. As observed in Fig.~\ref{fig:combinedstudy} (c), when the sources are uncorrelated, all the algorithms approach CRB as the SNR increases. When the sources are highly correlated and the SNR is high, $\vert\rho\vert=0.9$ in Fig.~\ref{fig:combinedstudy} (d), we observed that the performance curve for the proposed algorithm is the closest to CRB, followed by GL-SPICE. The performance curves are worse for RAM and GL-SPARROW indicating the effect of empirical bias present in the estimates.\vspace{-0.2cm}\begin{figure*}[h]
\centering
\begin{tabular}{ccc}
\includegraphics[width=0.3\linewidth]{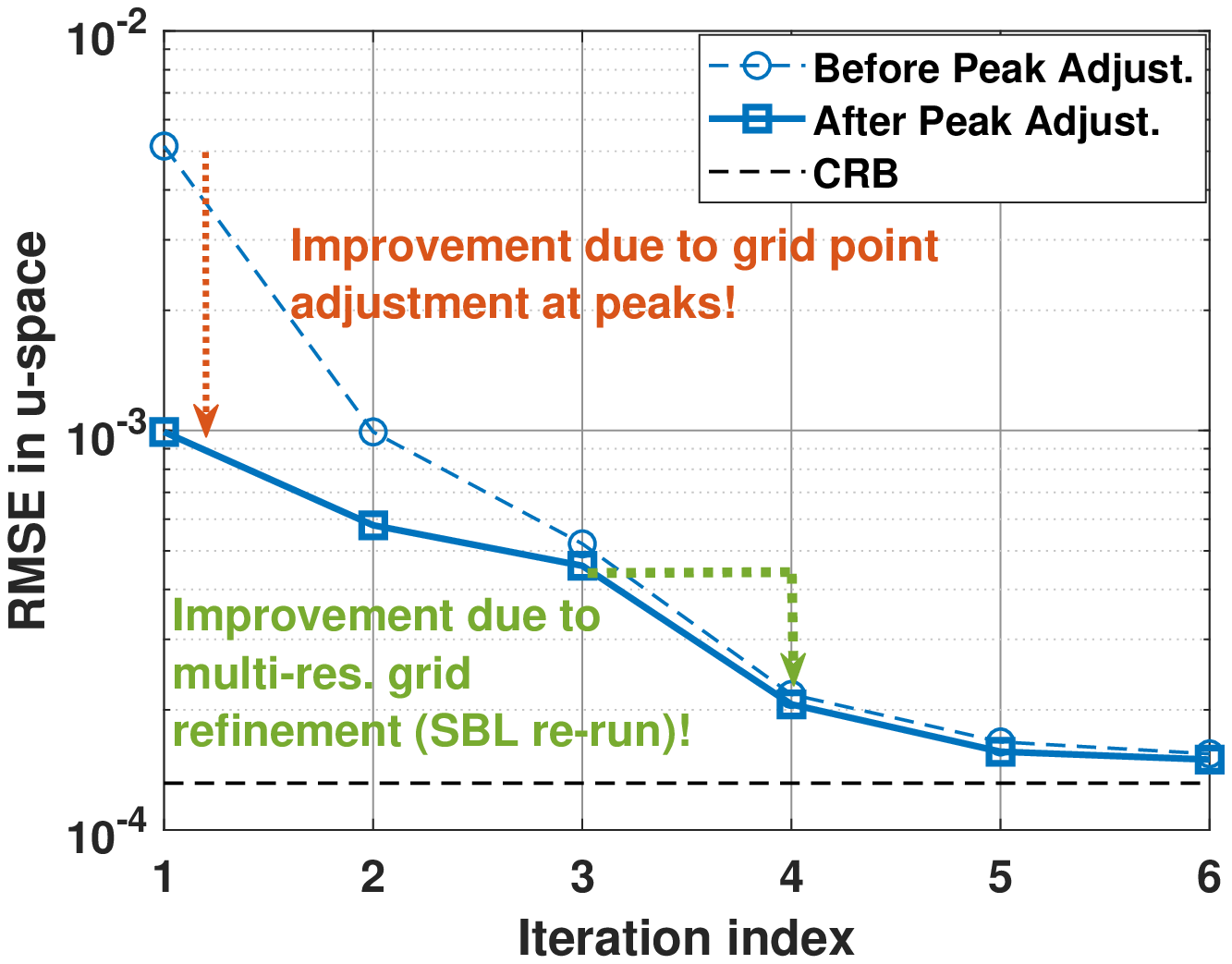} & \includegraphics[width=0.3\linewidth]{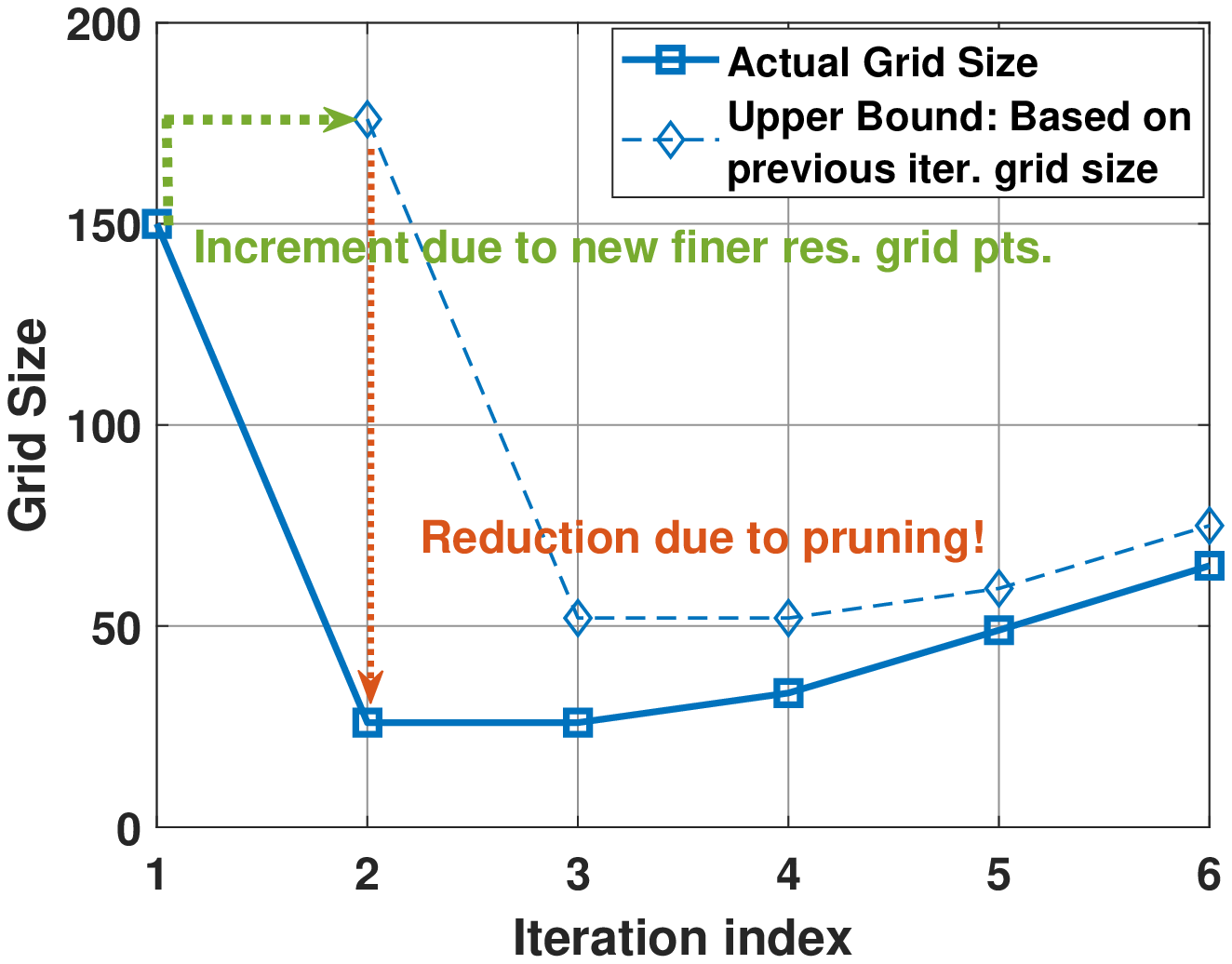} & \includegraphics[width=0.3\linewidth]{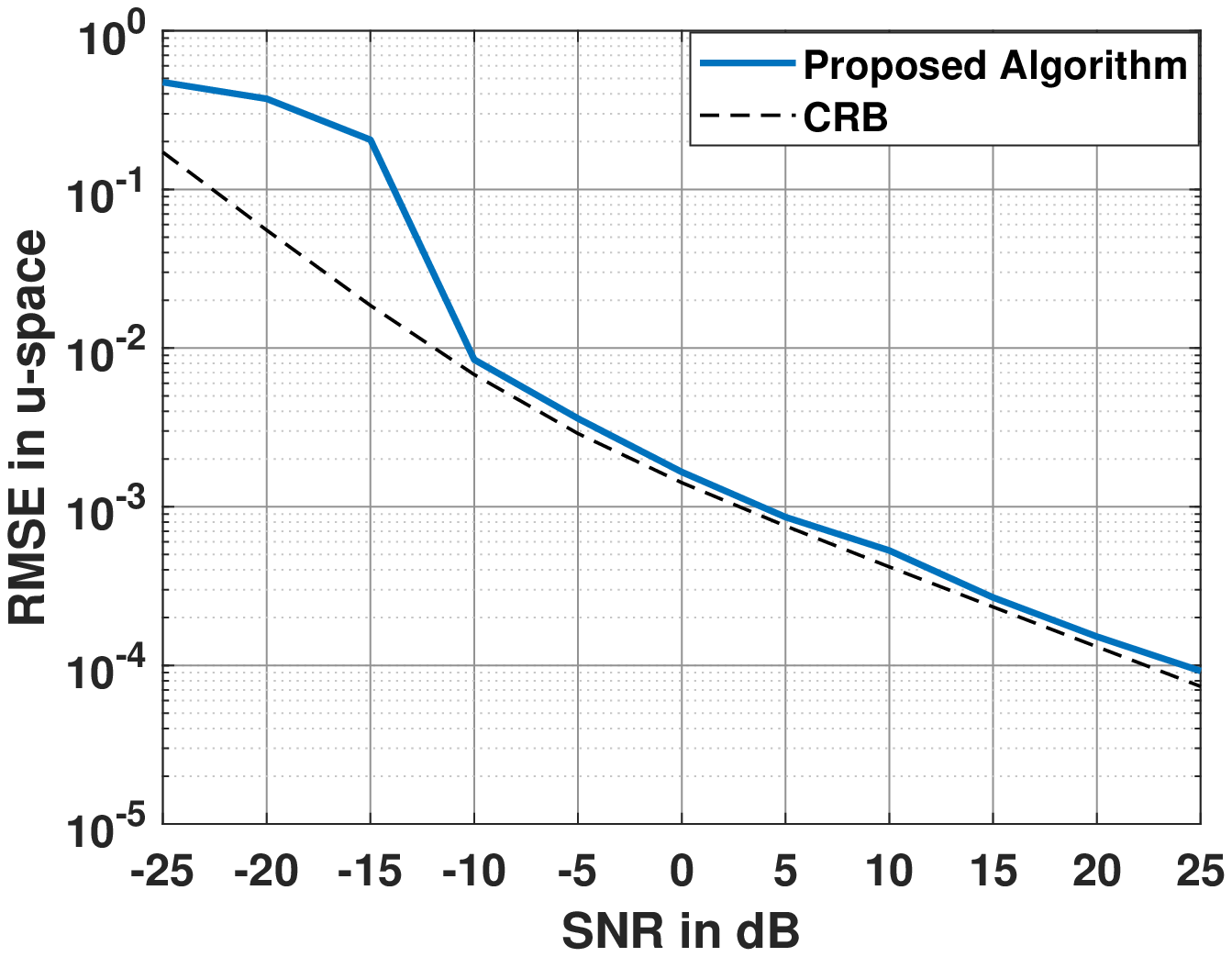}\\(a) RMSE over iterations ($r=0,1,\ldots$) & \begin{tabular}{c}(b) Grid size over iterations ($r=0,1,\ldots$)\\CRB req. grid size $=\frac{1}{\mbox{CRB}}\approx 7632$\end{tabular} & (c) RMSE vs. SNR
\end{tabular}
\caption{Non-uniform linear array: Performance of the proposed SBL with likelihood-based grid refinement procedure}\vspace{-0.5cm}\label{fig:arbitgeom}
\end{figure*}\subsection{Resolution Study and Regularization-free Proposed Approach vs RAM Study}\label{sec:combinedstudy}
\noindent We evaluate the performance of the proposed technique for resolution and compare it with RAM. We also compare the two algorithms for the case when sources have different SNRs. We consider a nested array with $M=10$ sensors, and allow $K=4$ sources incoming at angles $\{-0.5,-1/2M_{\mathrm{apt}},1/2M_{\mathrm{apt}},0.6\}$ in $u$-space, where $M_{\mathrm{apt}}=30$. The corresponding SNR for sources is $\{5,20,20,10\}$ dB and only a single snapshot ($L=1$) is available. The two sources near broadside are $1/M_{\mathrm{apt}}$ apart, or equivalently $0.5/M_{\mathrm{apt}}$ apart in normalized frequencies, which is a challenging scenario. As seen in Fig.~\ref{fig:combinedstudy} (e) and (f) for the proposed algorithm and RAM, respectively, both are able to resolve the two sources. The proposed algorithm is able to identify all $4$ sources, but RAM misses the weakest source. This behavior for RAM comes from the fact that the model is matched to an estimate of \emph{noiseless} data. In an attempt to construct such a noiseless estimate of measurement, the algorithm effectively suppressed the weakest source. It was observed that setting $\eta=0$ helped to identify all sources for RAM here. This indicates that RAM is sensitive to setting the parameter $\eta$ appropriately. Note also that the noise floor for the proposed algorithm is higher than that for RAM. This is expected because for the case of $K=4$ sources and $M=10$ sensors i.e., fewer sources than sensors, MUSIC is applied on the smaller $\mathbf{T}(\mathbf{v}^*)\in\mathbb{C}^{10\times 10}$ in the proposed algorithm, compared to MUSIC on $\T(\mathbf{v}^*)\in\mathbb{C}^{30\times 30}$ in RAM.\vspace{-0.2cm}\subsection{Performance of the Proposed SBL with Likelihood-based Grid Refinement Procedure
}\label{sec:gridrefinestudy}
\noindent We consider the array geometry mentioned in Section~\ref{sec:sblgdrefine} with $M=6$ sensors at positions $\mathbb{P}=\{0,1,2.1,3.5,4.7,10\}$, two sources at approximately $\{-0.5400,0.4802\}$ in $u$-space, and $L=500$. In Fig.~\ref{fig:arbitgeom} (a) and (b), we analyse the performance of the proposed algorithm at SNR=$20$ dB, while in (c) we consider a range of SNRs. We run $5000$ SBL iterations each time SBL is called, and the initial grid size is $G=150$. We plot the average results of $T=50$ random realizations in Fig.~\ref{fig:arbitgeom} (a) and (b) and $T=25$ random realizations in Fig.~\ref{fig:arbitgeom} (c).

\noindent In Fig.~\ref{fig:arbitgeom} (a), we plot the RMSE over the iterations ($r=0,1,\ldots$) described in Fig.~\ref{fig:twostepproc} for the proposed algorithm.
As observed here the RMSE decreases over iterations. Within each iteration (compare dashed vs. solid curve), it can be seen that the grid point adjustment step at peaks helps to reduce the error further, and thus establishes a simple way to further increase the likelihood. After $5$ iterations, it can be seen the error is very close to the CRB. In Fig.~\ref{fig:arbitgeom} (b), we plot the grid size (in solid blue curve) for running SBL at every iteration of the proposed procedure. In dashed blue curve, we plot a simple upper bound on the grid size by only counting the number of new grid points added around top peaks. This helps to compute the number of points pruned at every iteration. For reaching CRB at SNR=$20$ dB with a single SBL run and a uniform grid, we need a grid size of $G=7632$. In comparison, the proposed procedure only requires a maximum grid size of $G=150$. Note that the grid resolution around the top peaks, after $5$ iterations, is comparable to an initial grid size of $G\times g^(5-1)=150\times 3^4=12150$, which is more than enough to achieve CRB. In Fig.~\ref{fig:arbitgeom} (c), we plot the RMSE as a function of SNR. We set the maximum iterations ($<7$) of the proposed procedure so as to allow for sufficiently small grid spacing at high SNR as required to reach CRB. As seen in the plot, the RMSE approaches CRB as the SNR increases. Note that for a fixed grid size i.e., for a standard SBL procedure the RMSE is expected to saturate beyond a certain SNR.\vspace{-0.2cm}\section{Conclusion}\label{sec:conc}
\noindent In this work we revisited the problem of gridless sparse signal recovery
using MLE framework. We showed that SBL performs a structured covariance matrix estimation, where the structure is governed by the geometry of the measurement collection system (e.g. antenna array) and the (uncorrelated) source correlation prior.
We further established that SBL is a \emph{correlation-aware} technique and compared it with another class of correlation-aware techniques. Both are able to identify $O(M^2)$ sources given sparse linear array with $M$ sensors, like minimum-redundancy linear array\cite{moffet68,trees02,pillai85} and nested array. The noteworthy aspect about SBL is the underlying objective it uses, which is MLE. In the event that some of the sources are correlated, the model misfit is characterized in terms of the KL divergence between the distribution SBL assumes and the true data distribution. We reparametrized the SBL cost function to enable gridless support recovery when the sensors are placed on uniform grid and some sensors may be switched off. We provided an iterative algorithm based on linear MM to minimize the cost function and to estimate the structured covariance matrix of measurements. The DoAs can be recovered by using any off-the-shelf root-finding technique such as root-MUSIC. In this work, we also consider geometries when the sensors may be placed off the grid, and extend the SBL procedure to include a peak adjustment and grid refinement steps. Finally we compared the proposed algorithms numerically with other state-of-the-art algorithms from the literature and demonstrated the superior performance showcased by the cost function motivated by first principles, that is maximum likelihood estimation.

\noindent Several directions are open for future work. This includes, for all sensors on grid case, developing faster methods to solve the proposed `StructCovMLE' optimization problem. For the arbitrarily placed sensors' case, we feel the grid refinement based iterative SBL procedure is an important first step and opens up many interesting avenues of inquiry. For dictionaries parameterized by a  few parameters, there is hope that discretization (grid) may not be necessary upfront except as a practical computational method as in Equation (\ref{eq:neighopt2}).

\vspace{-0.2cm}\section{Appendix}\label{sec:apndx}
 \subsection{Proof of Proposition~\ref{prop:toepcon}}\label{sec:app_prop}
\begin{proof}
    The `\emph{if}' part can be proved simply using contradiction and follows by noting that if $(\mathbf{\Gamma}^{*},\lambda^{*})$ is not the global minimizer of (\ref{eq:toepcon}), then the solution for (\ref{eq:sblopt}) can be further improved.
We now prove the `\emph{only if}' part. If $\epsilon<\mathrm{tr}\left(\left(\mathbf{\Phi\Gamma}^*\mathbf{\Phi}^H+\lambda^{*}\mathbf{I}\right)^{-1}\hat{\mathbf{R}}_{\mathbf{y}}\right)$, then $(\mathbf{\Gamma}^{*},\lambda^{*})$ is infeasible, and the assertion holds trivially. If $\epsilon>\mathrm{tr}\left(\left(\mathbf{\Phi\Gamma}^*\mathbf{\Phi}^H+\lambda^{*}\mathbf{I}\right)^{-1}\hat{\mathbf{R}}_{\mathbf{y}}\right)$, then $(\mathbf{\Gamma}^{*},\lambda^{*})$ lies in the feasible region. We prove that the point $(\mathbf{\Gamma}^{*},\lambda^{*})$ can be further improved. For any two matrices $\mathbf{B},\mathbf{C}\succ\mathbf{0}$ such that $\mathbf{B}\succ\mathbf{C}$, the following holds\begin{IEEEeqnarray}{ll}
\mathrm{tr}\left((\mathbf{B}-\mathbf{C})^{-1}\hat{\mathbf{R}}_{\mathbf{y}}\right)\geq\mathrm{tr}\left(\mathbf{B}^{-1}\hat{\mathbf{R}}_{\mathbf{y}}\right)\\
\log\det\mathbf{B}>\log\det(\mathbf{B}-\mathbf{C}).\label{eq:improvsol}
\end{IEEEeqnarray}Inserting $\mathbf{B}=\mathbf{\Phi\Gamma}^*\mathbf{\Phi}^H+\lambda^{*}\mathbf{I}$ and $\mathbf{C}=\alpha\mathbf{I}$, for some $\alpha\in(0,\lambda^{*})$ in the above ensures that the conditions $\mathbf{B},\mathbf{C}\succ\mathbf{0}$ such that $\mathbf{B}\succ\mathbf{C}$ are satisfied. We choose $\alpha$ sufficiently small to ensure that the constraint $\mathrm{tr}\left((\mathbf{B}-\mathbf{C})^{-1}\hat{\mathbf{R}}_{\mathbf{y}}\right)\leq\epsilon$ is satisfied and consequently $(\mathbf{\Gamma}^{*},\lambda^{*}+\alpha)$ is feasible. Such an $\alpha$ exists because $\mathrm{tr}\left((\mathbf{B}-\mathbf{C})^{-1}\hat{\mathbf{R}}_{\mathbf{y}}\right)$ is a) continuous w.r.t. $\alpha$ in $(0,\lambda^{*})$ and b) right continuous at $\alpha=0$ with $\mathrm{tr}\left(\mathbf{B}^{-1}\hat{\mathbf{R}}_{\mathbf{y}}\right)<\epsilon$ as assumed. For such an $\alpha$, as evident from (\ref{eq:improvsol}), $(\mathbf{\Gamma}^{*},\lambda^{*}+\alpha)$ further improves the solution, and thus $(\mathbf{\Gamma}^{*},\lambda^{*})$ does not globally minimize (\ref{eq:toepcon}) if $\epsilon>\mathrm{tr}\left(\left(\mathbf{\Phi\Gamma}^*\mathbf{\Phi}^H+\lambda^{*}\mathbf{I}\right)^{-1}\hat{\mathbf{R}}_{\mathbf{y}}\right)$. This concludes the proof.\vspace{-0.3cm}
\end{proof}
\subsection{Proof of Theorem~\ref{thm:mleeqprop}}\label{prf:mleeqprop}
\begin{proof}The cost functions in (\ref{eq:mlcost3}) and (\ref{eq:mlcostKopt}) are identical, except for the received signal covariance matrix model. The optimization variables affect their cost only through the covariance matrix. Thus, the two problems are equivalent if the effective matrix search domains, up to an additional `$+\tilde{\lambda}\mathbf{I}$' ($\tilde{\lambda}\geq 0$) term, are same. Let $\mathcal{D}_1$ denote the matrix search region spanned by $\mathbf{T}(K,\bm{\theta},\mathbf{P})=\mathbf{\Phi}_{\bm{\theta}} \mathbf{P\Phi}_{\bm{\theta}}^H
$ in (\ref{eq:mlcostKopt}), and $\mathcal{D}_2$ for $\mathbf{T}(\mathbf{v})
$ in (\ref{eq:mlcost3}), where the domain for the parameters are indicated in the respective problems. To prove $\mathcal{D}_1\subseteq\mathcal{D}_2$: Let $\mathbf{T}(K',\bm{\theta}',\mathbf{P}')\in\mathcal{D}_1$ for some $(K',\bm{\theta}',\mathbf{P}')$, then the construction $\mathbf{v}'=\mathbf{T}^{-1}(\mathbf{\Phi}_{\bm{\theta}',\mathrm{ULA}}\mathbf{P}'\mathbf{\Phi}_{\bm{\theta}',\mathrm{ULA}}^H)$\footnote{$\mathbf{\Phi}_{\bm{\theta}',\mathrm{ULA}}$ denotes the array manifold matrix for a ULA of size $M_{\mathrm{apt}}$.
} ensures that $\T(\mathbf{v}')\succeq\mathbf{0}$ and $\mathbf{T}(\mathbf{v}')=\mathbf{T}(K',\bm{\theta}',\mathbf{P}')$, i.e., $\mathbf{T}(K',\bm{\theta}',\mathbf{P}')\in\mathcal{D}_2$. This concludes $\mathcal{D}_1\subseteq\mathcal{D}_2$. To prove $\mathcal{D}_2\subseteq\mathcal{D}_1$: Let $\mathbf{T}(\mathbf{v}'')\in\mathcal{D}_2$ for some $\mathbf{v}''$, then we have $\T(\mathbf{v}'')\succeq\mathbf{0}$. We skip the case when $\T(\mathbf{v}'')$ is low rank as it follows simply from unique Vandermonde decomposition. If $\T(\mathbf{v}'')$ is full rank, then it uniquely decomposes as $\mathbf{\Phi}_{\bm{\theta}'',\mathrm{ULA}}\mathbf{P}''\mathbf{\Phi}_{\bm{\theta}'',\mathrm{ULA}}^H+\lambda''\mathbf{I}$, for some $(\bm{\theta}'',\mathbf{P}'',\lambda''>0)$, where the corresponding $K''<M_{\mathrm{apt}}$ \cite{stoicabook05}. This ensures that $\mathbf{\Phi}_{\bm{\theta}''}\mathbf{P}''\mathbf{\Phi}_{\bm{\theta}''}^H+\lambda''\mathbf{I}=\mathbf{T}(\mathbf{v}'')$, which are equal up to the additional `$+\lambda''\mathbf{I}$' term.
This concludes that $\mathcal{D}_2\subseteq\mathcal{D}_1$.\vspace{-0.3cm}  
\end{proof}
\subsection{Proof of Theorem~\ref{thm:sbleqprop}}\label{prf:sbleqprop}
\begin{proof}
Similar to the proof for Theorem~\ref{thm:mleeqprop}, we conclude that the two problems in (\ref{eq:mlcost3}) and (\ref{eq:sbloptphiopt}) are equivalent if the effective matrix search domains, up to an additional `$+\tilde{\lambda}\mathbf{I}$' ($\tilde{\lambda}\geq 0$) term, are same. Let $\mathcal{D}_1$ denote the matrix search region spanned by $\mathbf{T}(\mathbf{\Phi},\bm{\Gamma})=\mathbf{\Phi}\bm{\Gamma}\mathbf{\Phi}^H
$ in (\ref{eq:sbloptphiopt}), and $\mathcal{D}_2$ for $\mathbf{T}(\mathbf{v})
$ in (\ref{eq:mlcost3}), where the domain for the parameters are indicated in the respective problems. To prove $\mathcal{D}_1\subseteq\mathcal{D}_2$: Let $\mathbf{T}(\mathbf{\Phi}',\bm{\Gamma}')\in\mathcal{D}_1$ for some $(\mathbf{\Phi}',\bm{\Gamma}')$, then the construction $\mathbf{v}'=\mathbf{T}^{-1}(\mathbf{\Phi}_{\mathrm{ULA}}'\bm{\Gamma}'\mathbf{\Phi}_{\mathrm{ULA}}'^H)$\footnote{$\mathbf{\Phi}_{\mathrm{ULA}}$ denotes the overcomplete dictionary for a ULA of size $M_{\mathrm{apt}}$, evaluated at same grid points as $\mathbf{\Phi}$.
} ensures that $\T(\mathbf{v}')\succeq\mathbf{0}$ and $\mathbf{T}(\mathbf{v}')=\mathbf{T}(\mathbf{\Phi}',\bm{\Gamma}')$, i.e., $\mathbf{T}(\mathbf{\Phi}',\bm{\Gamma}')\in\mathcal{D}_2$. This concludes $\mathcal{D}_1\subseteq\mathcal{D}_2$. To prove $\mathcal{D}_2\subseteq\mathcal{D}_1$: Let $\mathbf{T}(\mathbf{v}'')\in\mathcal{D}_2$ for some $\mathbf{v}''$, then we have $\T(\mathbf{v}'')\succeq\mathbf{0}$. Using Vandermonde decomposition we get $\T(\mathbf{v}'')=\mathbf{\Phi}_{\bm{\theta}'',\mathrm{ULA}}\mathbf{P}''\mathbf{\Phi}_{\bm{\theta}'',\mathrm{ULA}}^H$ for some $(\bm{\theta}'',\mathbf{P}''\succ\mathbf{0})$ which may not be unique. This decomposition leads to a valid dictionary $\mathbf{\Phi}''=\mathbf{\Phi}_{\bm{\theta}''}$ and diagonal source covariance matrix $\bm{\Gamma}''=\mathbf{P}''$. This concludes that $\mathcal{D}_2\subseteq\mathcal{D}_1$.\vspace{-0.3cm}
\end{proof}
\bibliographystyle{IEEEbib}
\bibliography{strings,refs}

\begin{thebibliography}{10}

\bibitem{gorodnitsky95}
I.F. Gorodnitsky, J.S. George, and B.D. Rao,
\newblock ``Neuromagnetic source imaging with {FOCUSS}: a recursive weighted
  minimum norm algorithm.,''
\newblock {\em J. Electroencephalog. Clinical Neurophysiol.}, vol. 95, no.4,
  pp. 231--251, 1995.

\bibitem{natarajan95}
B.~K. Natarajan,
\newblock ``Sparse {A}pproximate {S}olutions to {L}inear {S}ystems,''
\newblock {\em SIAM J. Comput.}, vol. 24, no. 2, pp. 227--234, 1995.

\bibitem{duttweiler00}
D.~L. Duttweiler,
\newblock ``Proportionate normalized least-mean-squares adaptation in echo
  cancelers,''
\newblock {\em IEEE Transactions on Speech and Audio Processing}, vol. 8, no.
  5, pp. 508--518, 2000.

\bibitem{stoicabook05}
P.~{Stoica} and R.~{Moses},
\newblock {\em Spectral Analysis of Signals},
\newblock Prentice-Hall, Upper Saddle River, NJ, USA, 2005.

\bibitem{krim96}
H.~{Krim} and M.~{Viberg},
\newblock ``Two decades of array signal processing research: the parametric
  approach,''
\newblock {\em IEEE Signal Processing Magazine}, vol. 13, no. 4, pp. 67--94,
  July 1996.

\bibitem{capon69}
J.~Capon,
\newblock ``High-resolution frequency-wavenumber spectrum analysis,''
\newblock {\em Proceedings of the IEEE}, vol. 57, no. 8, pp. 1408--1418, 1969.

\bibitem{schmidt86}
R.~Schmidt,
\newblock ``Multiple emitter location and signal parameter estimation,''
\newblock {\em IEEE Transactions on Antennas and Propagation}, vol. 34, no. 3,
  pp. 276--280, 1986.

\bibitem{roy89}
R.~Roy and T.~Kailath,
\newblock ``{ESPRIT}-estimation of signal parameters via rotational invariance
  techniques,''
\newblock {\em IEEE Transactions on Acoustics, Speech, and Signal Processing},
  vol. 37, no. 7, pp. 984--995, 1989.

\bibitem{mallat93}
S.~G. Mallat and Z.~Zhang,
\newblock ``Matching pursuits with time-frequency dictionaries,''
\newblock {\em IEEE Transactions on Signal Processing}, vol. 41, no. 12, pp.
  3397--3415, 1993.

\bibitem{chen98}
S.~S. Chen, D.~L. Donoho, and M.~A. Saunders,
\newblock ``Atomic decomposition by basis pursuit,''
\newblock {\em SIAM Journal on Scientific Computing}, vol. 20, no. 1, pp.
  33--61, 1998.

\bibitem{rao99}
B.~D. {Rao} and K.~{Kreutz-Delgado},
\newblock ``An affine scaling methodology for best basis selection,''
\newblock {\em IEEE Transactions on Signal Processing}, vol. 47, no. 1, pp.
  187--200, 1999.

\bibitem{tropp06}
J.A. Tropp,
\newblock ``Just relax: convex programming methods for identifying sparse
  signals in noise,''
\newblock {\em IEEE Transactions on Information Theory}, vol. 52, no. 3, pp.
  1030--1051, 2006.

\bibitem{malioutov05}
D.~Malioutov, M.~Cetin, and A.S. Willsky,
\newblock ``A sparse signal reconstruction perspective for source localization
  with sensor arrays,''
\newblock {\em IEEE Transactions on Signal Processing}, vol. 53, no. 8, pp.
  3010--3022, 2005.

\bibitem{wipf04}
D.~{Wipf} and B.~D. {Rao},
\newblock ``Sparse {B}ayesian learning for basis selection,''
\newblock {\em IEEE Transactions on Signal Processing}, vol. 52, no. 8, pp.
  2153--2164, Aug 2004.

\bibitem{cotter05}
S.F. Cotter, B.D. Rao, K.~Engan, and K.~Kreutz-Delgado,
\newblock ``Sparse solutions to linear inverse problems with multiple
  measurement vectors,''
\newblock {\em IEEE Trans. on Signal Processing}, vol. 53, no. 7, pp.
  2477--2488, 2005.

\bibitem{tang13}
G.~{Tang}, B.~N. {Bhaskar}, P.~{Shah}, and B.~{Recht},
\newblock ``Compressed sensing off the grid,''
\newblock {\em IEEE Transactions on Information Theory}, vol. 59, no. 11, pp.
  7465--7490, 2013.

\bibitem{steffens18}
C.~{Steffens}, M.~{Pesavento}, and M.~E. {Pfetsch},
\newblock ``A compact formulation for the $\ell _{2,1}$ mixed-norm minimization
  problem,''
\newblock {\em IEEE Transactions on Signal Processing}, vol. 66, no. 6, pp.
  1483--1497, 2018.

\bibitem{yang16}
Z.~{Yang} and L.~{Xie},
\newblock ``Enhancing sparsity and resolution via reweighted atomic norm
  minimization,''
\newblock {\em IEEE Transactions on Signal Processing}, vol. 64, no. 4, pp.
  995--1006, 2016.

\bibitem{yang15}
Z.~{Yang} and L.~{Xie},
\newblock ``On gridless sparse methods for line spectral estimation from
  complete and incomplete data,''
\newblock {\em IEEE Transactions on Signal Processing}, vol. 63, no. 12, pp.
  3139--3153, 2015.

\bibitem{tipping01}
M.~Tipping,
\newblock ``Sparse {B}ayesian learning and the relevance vector machine,''
\newblock {\em Machine Learning Research}, vol. 1, pp. 211--244, 2001.

\bibitem{wipfrao07}
D.~P. {Wipf} and B.~D. {Rao},
\newblock ``An empirical {B}ayesian strategy for solving the simultaneous
  sparse approximation problem,''
\newblock {\em IEEE Transactions on Signal Processing}, vol. 55, no. 7, pp.
  3704--3716, 2007.

\bibitem{wipf10}
D.~Wipf and S.~Nagarajan,
\newblock ``Iterative reweighted $\ell_1$ and $\ell_2$ methods for finding
  sparse solutions,''
\newblock {\em IEEE Journal of Selected Topics in Signal Processing}, vol. 4,
  no. 2, pp. 317--329, 2010.

\bibitem{pal15}
P.~{Pal} and P.~P. {Vaidyanathan},
\newblock ``Pushing the limits of sparse support recovery using correlation
  information,''
\newblock {\em IEEE Transactions on Signal Processing}, vol. 63, no. 3, pp.
  711--726, 2015.

\bibitem{balkan14}
O.~{Balkan}, K.~{Kreutz-Delgado}, and S.~{Makeig},
\newblock ``Localization of more sources than sensors via jointly-sparse
  {B}ayesian learning,''
\newblock {\em IEEE Signal Processing Letters}, vol. 21, no. 2, pp. 131--134,
  2014.

\bibitem{sun17}
Y.~Sun, P.~Babu, and D.~P. Palomar,
\newblock ``Majorization-minimization algorithms in signal processing,
  communications, and machine learning,''
\newblock {\em IEEE Trans. on Signal Processing}, vol. 65, no. 3, pp. 794--816,
  2017.

\bibitem{lipp16}
T.~{Lipp} and S.~{Boyd},
\newblock ``Variations and extension of the convex–concave procedure,''
\newblock {\em Optimization and Engineering}, vol. 17, no. 9, pp. 1573--2924,
  2016.

\bibitem{pote22}
R.~R. Pote and B.~D. Rao,
\newblock ``Maximum likelihood structured covariance matrix estimation and
  connections to {SBL}: A path to gridless {D}o{A} estimation,''
\newblock in {\em 2022 56th Asilomar Conference on Signals, Systems, and
  Computers}, 2022.

\bibitem{wassim17}
W.~Suleiman, C.~Steffens, A.~Sorg, and M.~Pesavento,
\newblock ``Gridless compressed sensing for fully augmentable arrays,''
\newblock in {\em 2017 25th European Signal Processing Conference (EUSIPCO)},
  2017, pp. 1986--1990.

\bibitem{burg82}
J.~P. {Burg}, D.~G. {Luenberger}, and D.~L. {Wenger},
\newblock ``Estimation of structured covariance matrices,''
\newblock {\em Proceedings of the IEEE}, vol. 70, no. 9, pp. 963--974, 1982.

\bibitem{bohme86}
J.~{Bohme},
\newblock ``Separated estimation of wave parameters and spectral parameters by
  maximum likelihood,''
\newblock in {\em ICASSP '86. IEEE International Conference on Acoustics,
  Speech, and Signal Processing}, 1986, vol.~11, pp. 2819--2822.

\bibitem{miller87}
M.~I. {Miller} and D.~L. {Snyder},
\newblock ``The role of likelihood and entropy in incomplete-data problems:
  Applications to estimating point-process intensities and toeplitz constrained
  covariances,''
\newblock {\em Proceedings of the IEEE}, vol. 75, no. 7, pp. 892--907, 1987.

\bibitem{jaffer88}
A.~G. {Jaffer},
\newblock ``Maximum likelihood direction finding of stochastic sources: a
  separable solution,''
\newblock in {\em ICASSP-88., International Conference on Acoustics, Speech,
  and Signal Processing}, 1988, pp. 2893--2896 vol.5.

\bibitem{fuhrmann88}
D.~R. {Fuhrmann},
\newblock ``Progress in structured covariance estimation,''
\newblock in {\em Fourth Annual ASSP Workshop on Spectrum Estimation and
  Modeling}, 1988, pp. 158--161.

\bibitem{li99}
{H. Li}, P.~{Stoica}, and {J. Li},
\newblock ``Computationally efficient maximum likelihood estimation of
  structured covariance matrices,''
\newblock {\em IEEE Transactions on Signal Processing}, vol. 47, no. 5, pp.
  1314--1323, 1999.

\bibitem{stoica95}
P.~{Stoica} and A.~{Nehorai},
\newblock ``On the concentrated stochastic likelihood function in array signal
  processing,''
\newblock {\em Circuits, Systems and Signal Process.}, vol. 14, pp. 669--674,
  1995.

\bibitem{viberg97}
M.~{Viberg}, P.~{Stoica}, and B.~{Ottersten},
\newblock ``Maximum likelihood array processing in spatially correlated noise
  fields using parameterized signals,''
\newblock {\em IEEE Trans. on Signal Processing}, vol. 45, no. 4, pp.
  996--1004, 1997.

\bibitem{ottersten98}
B~Ottersten, P~Stoica, and R~Roy,
\newblock ``Covariance matching estimation techniques for array signal
  processing applications,''
\newblock {\em Digital Signal Processing}, vol. 8, no. 3, pp. 185 -- 210, 1998.

\bibitem{chandrasekaran12}
V.~{Chandrasekaran}, B.~{Recht}, P.~A. {Parrilo}, and A.~S. {Willsky},
\newblock ``The convex geometry of linear inverse problems,''
\newblock {\em Foundations of Computational Mathematics}, vol. 12, pp.
  805--849, 2012.

\bibitem{candes14}
E.~J. Candès and C.~Fernandez-Granda,
\newblock ``Towards a mathematical theory of super-resolution,''
\newblock {\em Communications on Pure and Applied Mathematics}, vol. 67, no. 6,
  pp. 906--956, 2014.

\bibitem{stoica11}
P.~{Stoica}, P.~{Babu}, and J.~{Li},
\newblock ``{SPICE}: A sparse covariance-based estimation method for array
  processing,''
\newblock {\em IEEE Transactions on Signal Processing}, vol. 59, no. 2, pp.
  629--638, 2011.

\bibitem{stoica12}
P.~Stoica and P.~Babu,
\newblock ``{SPICE} and {LIKES}: Two hyperparameter-free methods for
  sparse-parameter estimation,''
\newblock {\em Signal Processing}, vol. 92, no. 7, pp. 1580 -- 1590, 2012.

\bibitem{pal12}
P.~Pal and P.~P. Vaidyanathan,
\newblock ``Correlation-aware techniques for sparse support recovery,''
\newblock in {\em 2012 IEEE Statistical Signal Processing Workshop (SSP)},
  2012, pp. 53--56.

\bibitem{trees02}
H.~L.~Van Trees,
\newblock ``Optimum {A}rray {P}rocessing: {P}art {IV} of {D}etection,
  {E}stimation, and {M}odulation {T}heory,''
\newblock {\em John Wiley \& Sons, Ltd}, p. 178, 2002.

\bibitem{tipping03}
M.~E. Tipping and A.~C. Faul,
\newblock ``Fast marginal likelihood maximisation for sparse {B}ayesian
  models,''
\newblock in {\em Proceedings of the Ninth International Workshop on Artificial
  Intelligence and Statistics}. 03--06 Jan 2003, vol.~R4, pp. 276--283, PMLR.

\bibitem{alshoukairi18}
M.~Al-Shoukairi, P.~Schniter, and B.~D. Rao,
\newblock ``A {GAMP}-based low complexity sparse {B}ayesian learning
  algorithm,''
\newblock {\em IEEE Transactions on Signal Processing}, vol. 66, no. 2, pp.
  294--308, 2018.

\bibitem{shengheng19}
S.~Liu, H.~Wu, Y.~Huang, Y.~Yang, and J.~Jia,
\newblock ``Accelerated structure-aware sparse {B}ayesian learning for
  three-dimensional electrical impedance tomography,''
\newblock {\em IEEE Transactions on Industrial Informatics}, vol. 15, no. 9,
  pp. 5033--5041, 2019.

\bibitem{wipf08}
D.~Wipf and S.~Nagarajan,
\newblock ``A new view of automatic relevance determination,''
\newblock in {\em Advances in Neural Information Processing Systems}. 2008,
  vol.~20, Curran Associates, Inc.

\bibitem{qiao19}
H.~{Qiao} and P.~{Pal},
\newblock ``Guaranteed localization of more sources than sensors with finite
  snapshots in multiple measurement vector models using difference co-arrays,''
\newblock {\em IEEE Transactions on Signal Processing}, vol. 67, no. 22, pp.
  5715--5729, 2019.

\bibitem{wipf11}
D.~P. Wipf, B.~D. Rao, and S.~Nagarajan,
\newblock ``Latent variable {B}ayesian models for promoting sparsity,''
\newblock {\em IEEE Transactions on Information Theory}, vol. 57, no. 9, pp.
  6236--6255, 2011.

\bibitem{barabell83}
A.~{Barabell},
\newblock ``Improving the resolution performance of eigenstructure-based
  direction-finding algorithms,''
\newblock in {\em IEEE International Conference on Acoustics, Speech, and
  Signal Processing}, 1983, vol.~8, pp. 336--339.

\bibitem{caratheodoryfejer11}
C.~{Carath\'eodory} and L.~{Fej\'er},
\newblock ``{Uber den {Z}usammenhang der {E}xtremen von harmonischen
  {F}unktionen mit ihren {K}oeffizienten und uber den {P}icard-{L}andau’schen
  {S}atz},''
\newblock {\em Rendiconti del Circolo Matematico di Palermo (1884-1940)}, vol.
  32, no. 1, pp. 218–239, 1911.

\bibitem{fazel03}
M.~{Fazel}, H.~{Hindi}, and S.~P. {Boyd},
\newblock ``Log-det heuristic for matrix rank minimization with applications to
  hankel and euclidean distance matrices,''
\newblock in {\em Proceedings of the 2003 American Control Conference, 2003.},
  2003, vol.~3, pp. 2156--2162 vol.3.

\bibitem{pote20}
R.~R. {Pote} and B.~D. {Rao},
\newblock ``Robustness of sparse {B}ayesian learning in correlated
  environments,''
\newblock in {\em IEEE International Conference on Acoustics, Speech and Signal
  Processing}, 2020, pp. 9100--9104.

\bibitem{wipf07}
D.~Wipf and S.~Nagarajan,
\newblock ``Beamforming using the relevance vector machine,''
\newblock in {\em Proceedings of the 24th International Conference on Machine
  Learning}, New York, NY, USA, 2007, ICML '07, p. 1023–1030, Association for
  Computing Machinery.

\bibitem{boyd04}
Stephen Boyd and Lieven Vandenberghe,
\newblock {\em Convex Optimization},
\newblock Cambridge University Press, 2004.

\bibitem{cvx}
M.~Grant and S.~Boyd,
\newblock ``{CVX}: Matlab software for disciplined convex programming, version
  2.1,'' http://cvxr.com/cvx, Mar. 2014.

\bibitem{fazel01}
M.~{Fazel}, H.~{Hindi}, and S.~P. {Boyd},
\newblock ``A rank minimization heuristic with application to minimum order
  system approximation,''
\newblock in {\em Proceedings of the 2001 American Control Conference. (Cat.
  No.01CH37148)}, 2001, vol.~6, pp. 4734--4739 vol.6.

\bibitem{pillai85}
S.~U. {Pillai}, Y.~{Bar-Ness}, and F.~{Haber},
\newblock ``A new approach to array geometry for improved spatial spectrum
  estimation,''
\newblock {\em Proceedings of the IEEE}, vol. 73, no. 10, pp. 1522--1524, 1985.

\bibitem{stoica90}
P.~Stoica and A.~Nehorai,
\newblock ``Performance study of conditional and unconditional
  direction-of-arrival estimation,''
\newblock {\em IEEE Trans. on Acoustics, Speech, and Signal Processing}, vol.
  38, no. 10, pp. 1783--1795, 1990.

\bibitem{qiao17}
H.~Qiao and P.~Pal,
\newblock ``On maximum-likelihood methods for localizing more sources than
  sensors,''
\newblock {\em IEEE Signal Processing Letters}, vol. 24, no. 5, pp. 703--706,
  2017.

\bibitem{yang22}
Z.~Yang, X.~Chen, and X.~Wu,
\newblock ``A robust and statistically efficient maximum-likelihood method for
  doa estimation using sparse linear arrays,''
\newblock {\em arXiv}, 2022.

\bibitem{alshoukairi21}
M.~Al\-Shoukairi,
\newblock ``Message passing algorithms and extensions of sparse {B}ayesian
  learning,''
\newblock {\em UC San Diego. ProQuest ID: AlShoukairi\_ucsd\_0033D\_20301.
  Merritt ID: ark:/13030/m52g3pxb. Retrieved from
  https://escholarship.org/uc/item/88812150}, 2021.

\bibitem{stoica01}
P.~Stoica, E.G. Larsson, and A.B. Gershman,
\newblock ``The stochastic {CRB} for array processing: a textbook derivation,''
\newblock {\em IEEE Signal Processing Letters}, vol. 8, no. 5, pp. 148--150,
  2001.

\bibitem{pal10}
P.~{Pal} and P.~P. {Vaidyanathan},
\newblock ``Nested arrays: A novel approach to array processing with enhanced
  degrees of freedom,''
\newblock {\em IEEE Transactions on Signal Processing}, vol. 58, no. 8, pp.
  4167--4181, 2010.

\bibitem{moffet68}
A.~{Moffet},
\newblock ``Minimum-redundancy linear arrays,''
\newblock {\em IEEE Transactions on Antennas and Propagation}, vol. 16, no. 2,
  pp. 172--175, 1968.

\end{thebibliography}
\end{document}